\documentclass[runningheads]{llncs}
\usepackage[T1]{fontenc}
\usepackage{graphicx}
\usepackage{color}

\usepackage{mathtools,amsmath,amssymb,amsfonts}
\usepackage{stmaryrd}
\usepackage{hyperref}
\usepackage{xspace}

\usepackage{algorithm}
\usepackage[noend]{algpseudocode}
\usepackage{changepage}
\usepackage{thm-restate}
\usepackage{multirow}
\usepackage{cleveref}
\crefname{definition}{Def.}{Def.}
\crefname{equation}{}{}
\crefname{figure}{Fig.}{Fig.}
\crefname{problem}{Prob.}{Prob.}
\crefname{algorithm}{Algo.}{Algo.}
\crefname{theorem}{Thm.}{Thm.}
\crefname{lemma}{Lem.}{Lem.}
\crefname{appendix}{App.}{App.}
\crefname{line}{line}{lines}
\crefname{item}{item}{items}
\crefname{section}{Sec.}{Sec.}
\crefname{corollary}{Cor.}{Cor.}
\crefname{proposition}{Prop.}{Prop.}
\crefname{example}{Ex.}{Ex.}
\crefname{remark}{Rem.}{Rem.}
\usepackage{graphicx}
\usepackage{subcaption}
\usepackage[table,xcdraw]{xcolor}
\usepackage{longtable}
\usepackage{paralist}

\usepackage[
colorinlistoftodos, %
textwidth=\marginparwidth, %
textsize=scriptsize, %
]{todonotes}

\usepackage{mdframed}

\usepackage{pgf}
\usepackage{tikz}
\usetikzlibrary{shapes,shapes.geometric,fit,trees,decorations,arrows,automata,shadows,positioning,plotmarks,backgrounds,tikzmark}
\usetikzlibrary{calc,matrix,fit,petri,decorations.markings,decorations.pathmorphing,patterns,intersections,decorations.text}

\usepackage{graphicx}
\usepackage[utf8x]{inputenc}
\usepackage[american]{babel}
\usepackage[babel]{csquotes}

\usepackage{algorithm, algpseudocode}

\usepackage{amscd}
\usepackage{amsmath}
\usepackage{amssymb}
\usepackage{mathtools}
\usepackage{array}

\usepackage{xspace}
\usepackage{paralist}
\usepackage{xifthen}
\usepackage{url}

\usepackage{pgf}
\usepackage{tikz}
\usepackage{lmodern}

\usepackage{MnSymbol,wasysym}
\usepackage{comment}

\colorlet{darkgreen}{green!80!black}
\colorlet{darkred}{red!80!black}
\usetikzlibrary{arrows, automata, shapes}
\tikzset{auto, >= stealth}
\tikzset{every edge/.append style={thick, shorten >= 1pt}}
\tikzset{initial/.style={draw, thick, <-, shorten <=1pt}}
\tikzset{player2/.style = {draw, thick, shape=circle, minimum size=5mm}}
\tikzset{player1/.style = {draw, thick, shape=rectangle, minimum size=5mm}}
\tikzset{player3/.style = {draw, thick, shape=diamond, minimum size=5mm}}
\tikzset{bplayer0/.style = {draw, thick, shape=ellipse, minimum size=3mm,text width=0.65cm}}

\newcommand\hpos{1.8}
\newcommand\ypos{1.4}

\newcommand{\bigO}{\mathcal{O}}

\newcommand{\inodd}{\in_{\text{odd}}}
\newcommand{\ineven}{\in_{\text{even}}}
\newcommand{\abs}[1]{\left\lvert #1 \right\rvert}
\newcommand{\set}[1]{\left\lbrace #1\right\rbrace}
\newcommand{\tup}[1]{\left( #1\right)}

\newcommand{\false}{\mathsf{False}}
\newcommand{\true}{\mathsf{True}}
\newcommand{\AP}{\mathtt{AP}}
\newcommand{\LTLnext}{\bigcirc}
\newcommand{\LTLeventually}{\lozenge}
\newcommand{\LTLalways}{\square}
\newcommand{\LTLuntil}{\ensuremath{~\mathcal{U}~}}
\newcommand{\lang}{\mathcal{L}}

\newcommand{\game}{\mathcal{G}}
\newcommand{\gamegraph}{\ensuremath{G}}
\newcommand{\vertex}{V}
\newcommand{\vertexi}{\vertex_i}

\newcommand{\play}{\ensuremath{\rho}}
\newcommand{\playprefix}{\play}
\newcommand{\priority}{\Omega}
\newcommand{\priorityset}[1]{\priority_{#1}}

\newcommand{\spec}{\ensuremath{\phi}}
\newcommand{\speci}{\spec_i}
\newcommand{\specj}{\spec_j}
\newcommand{\specnoti}{\spec_{-i}}
\newcommand{\specSE}{\varphi}
\newcommand{\specSEi}{\specSE_i}

\newcommand{\specSEnoti}{\specSE_{-i}}

\newcommand{\p}[1]{\ensuremath{\text{Player}~#1}}
\newcommand{\players}{\mathtt{P}}
\newcommand{\playersnoti}{\players_{-i}}
\newcommand{\env}{\mathtt{env}}

\newcommand{\win}{\mathcal{W}}

\newcommand{\team}[1]{\llangle #1\rrangle}
\newcommand{\teamall}{\team{\players}}

\newcommand{\strat}{\ensuremath{\pi}}
\newcommand{\strati}{\ensuremath{\strat_{i}}}
\newcommand{\stratj}{\ensuremath{\strat_{j}}}
\newcommand{\stratnoti}{\ensuremath{\strat_{-i}}}

\newcommand{\Strat}{\ensuremath{\Pi}}

\newcommand{\buchi}{\ifmmode B\ddot{u}chi \else B\"uchi \fi}
\newcommand{\cobuchi}{\ifmmode co\text{-}B\ddot{u}chi \else co-B\"uchi \fi}

\newcommand{\parity}{\ensuremath{\textit{Parity}}}

\newcommand{\colivegroup}{C}
\newcommand{\safegroup}{S}
\newcommand{\assump}{\ensuremath{\psi}}
\newcommand{\assumpi}{\assump_i}

\newcommand{\assumpnoti}{\assump_{-i}}
\newcommand{\assumptemp}[2]{\assump^{[#1,#2]}}
\newcommand{\assumpsafe}{\ensuremath{\assump_{\textsc{unsafe}}}}

 \newcommand{\assumpgrlive}{\ensuremath{\assump_{\textsc{cond}}}}
 \newcommand{\assumpcolive}{\ensuremath{\assump_{\textsc{colive}}}}

\newcommand{\computeGE}{\textsc{ComputeGE}}
\newcommand{\ocomputeGE}{\textsc{OComputeGE}}
\newcommand{\computeAPA}{\textsc{ComputeAPA}}
\newcommand{\computeUCA}{\textsc{ComputeUCA}}

\newcommand{\computeWin}{\textsc{ComputeWin}}
\newcommand{\recursiveGE}{\textsc{RecursiveGE}}

\newcommand{\approxAPA}{\textsc{ApproxAPA}}

\newcommand{\payoff}{\mathsf{payoff}}

\begin{document}
\title{Most General Winning Secure Equilibria Synthesis in Graph Games\thanks{Authors are supported by the DFG project 389792660 TRR 248-CPEC. 
Additionally A.-K.~Schmuck is supported by the DFG project SCHM 3541/1-1.}}

\author{Satya Prakash Nayak \and Anne-Kathrin Schmuck}
\authorrunning{S. P. Nayak and A. Schmuck}
\institute{Max Planck Institute for Software Systems, Kaiserslautern, Germany\\
\email{\{sanayak,akschmuck\}@mpi-sws.org}}

\maketitle              %
\begin{abstract}
This paper considers the problem of co-synthesis in $k$-player games over a finite graph where each player has an individual $\omega$-regular specification $\phi_i$. In this context, a secure equilibrium (SE) is a Nash equilibrium w.r.t.\ the lexicographically ordered objectives of each player to first satisfy their own specification, and second, to falsify other players' specifications.
A \emph{winning secure equilibrium} (WSE) is an SE strategy profile $(\pi_i)_{i\in[1;k]}$ that ensures the specification $\phi:=\bigwedge_{i\in[1;k]}\phi_i$ if no player deviates from their strategy $\pi_i$. Distributed implementations generated from a WSE make components \emph{act rationally} by ensuring that a deviation from the WSE strategy profile is \emph{immediately punished} by a \emph{retaliating strategy} that makes the involved players lose.

In this paper, we move from \emph{deviation punishment} in WSE-based implementations to a \emph{distributed, assume-guarantee based realization} of WSE. This shift is obtained by generalizing WSE from \emph{strategy profiles} to \emph{specification profiles} $(\varphi_i)_{i\in[1;k]}$ with $\bigwedge_{i\in[1;k]}\varphi_i = \phi$, which we call \emph{most general winning secure equilibria (GWSE)}. Such GWSE have the property that each player can individually pick a strategy $\pi_i$ winning for $\varphi_i$ (against all other players) and all resulting strategy profiles $(\pi_i)_{i\in[1;k]}$ are guaranteed to be a WSE. The obtained flexibility in players' strategy choices can be utilized for robustness and adaptability of local implementations. 

Concretely, our contribution is three-fold: (1) we formalize GWSE for $k$-player games over finite graphs, where each player has an $\omega$-regular specification $\phi_i$; (2) we devise an \emph{iterative semi-algorithm} for GWSE synthesis in such games, and (3) obtain an \emph{exponential-time algorithm} for GWSE synthesis with parity specifications $\phi_i$.
	
\end{abstract}

\section{Introduction}\label{section:intro}

Games over graphs provide a well known abstraction for many challenging correct-by-construction synthesis problems for software and hardware in embedded cyber-physical applications. In particular, the correct-by-construction co-synthesis of multiple interacting (reactive) components -- each with its own correctness specification -- poses, as of today, severe challenges in automated system design. 

While many of these challenges arise from the fact that not every component has the same information about all relevant variables in the system, even in the seemingly simple setting of \emph{full information} -- where all components see the valuation to all variables -- finding the right balance between centralized and local reasoning for co-synthesis is surprisingly challenging. While assuming all players to cooperate might demand too much commitment from individual components, a fully adversarial setting where all other components are assumed to harm a local implementation (independently of their own objective) might not capture a realistic scenario either. 

To address this issue, starting with the seminal work of Chatterjee et al.~\cite{chatterjee_SE}, the concept of \emph{rationality} -- stemming from classical game theory -- was brought to graph games in order to formalize a more realistic model for interaction of multiple components in co-synthesis. The main conceptual contribution of \cite{chatterjee_SE} was the introduction of \emph{secure equilibria} (SE) -- a special sub-class of Nash equilibria -- given as particular strategy profiles. Intuitively, an SE is a Nash equilibrium w.r.t.\ the lexicographically ordered objectives of each player to first satisfy their own specification, and only second, to falsify other players' specifications. More specifically, it is a strategy profile, i.e., a tuple $(\strati)_{i}$, with $\strati$ being the strategy of $\p{i}$, such that no player can improve w.r.t.\ their lexicographically ordered objective by deviating from this strategy.

As stated by \cite[p.68]{chatterjee_SE}, an SE can thus be interpreted as a contract between the players which enforces cooperation: any unilateral selfish deviation by one player cannot put the other players at a disadvantage if they follow the SE. While this property makes SE very desirable, their main draw-back, as most prominently pointed out by \cite{admissible}, is their restriction to a \emph{single} strategy profile. This, in combination with classical reactive synthesis engines typically preferring small and goal-oriented strategies, incentivizes \enquote{immediate punishment} of deviations from an SE strategy profile in the final implementation.

\smallskip
\noindent\textbf{Motivating Example.}
\begin{figure}[b]
\centering
\vspace{-0.5cm}
\begin{tikzpicture}
	\node[player2] (v0) at (0, 0) {$v_0$};
    \node[player2,accepting] (v2) at (\hpos,0) {$v_2$};
    \node[player1,accepting] (v1) at (0, -\ypos) {$v_1$};
    \node[player1] (v3) at (\hpos, -\ypos) {$v_3$};
    \node[player1] (v4) at (2*\hpos, -\ypos) {$v_4$};
                    
    \path[->] (v2) edge[bend left =20] (v3) edge (v4);
    \path[->] (v0) edge (v2) edge (v3);
    \path[->] (v3) edge[bend left =20] (v2) edge (v1) edge (v4);
    \path[->] (v4) edge[loop right] ();
    \path[->] (v1) edge (v0);
\end{tikzpicture}
\caption{A two-player game with $\p{1}$'s vertices (squares), $\p{2}$'s vertices (circles) where $\p{i}$'s specification $\speci = \LTLalways\LTLeventually v_i$ is to visit $v_i$ infinitely often.}
\vspace{-1cm}
\label{fig:motivating-intro}
\end{figure}
To illustrate this effect, let us consider the game depicted in \cref{fig:motivating-intro}, taken from \cite{chatterjee_SE}.
Here, an SE can be described as follows: if $\p{1}$ always chooses $v_3\rightarrow v_1$ (forming $\strat_1$) and $\p{2}$ always chooses $v_0\rightarrow v_2$ and $v_2\rightarrow v_3$ (forming $\strat_2$), then they both satisfy their specifications; if $\p{1}$ deviates by choosing $v_3\rightarrow v_2$ (risking falsification of $\spec_2$), then $\p{2}$ can retaliate by choosing $v_2\rightarrow v_4$ (ensuring falsification of both $\spec_i$); similarly, if $\p{2}$ deviates by choosing $v_0\rightarrow v_3$ (risking falsification of $\spec_1$), then $\p{1}$ retaliate by choosing $v_3\rightarrow v_4$ (ensuring falsification of both $\spec_i$).
Clearly, the strategy profile $(\strat_1,\strat_2)$ is an SE. It is, in particular, a \emph{winning} SE as both players satisfy their specifications when following it. However, as the outlined retaliating strategies $(\strat'_1,\strat'_2)$ are also part of the final implementation generated from this SE, any play that deviates from $(\strat_1,\strat_2)$ \emph{only once}, makes the game end up in a loop at $v_4$ resulting in neither player satisfying their objectives. Intuitively, this way of implementing SE-based strategies makes components \emph{act rationally} by ensuring that a deviation from the contract is \emph{immediately punished}.

Having the interpretation of an SE as a contract in mind, it is however very appealing to think about the \emph{realization} of this contract in the final implementation in a more \emph{permissive} way. 
Intuitively, in the game depicted in \cref{fig:motivating-intro}, both players can satisfy their specifications $\spec_i$ without the help by the other player, as long as the play does not go to $v_4$. In particular, whenever both players independently chose a strategy $\pi_i$ which ensures that they (i) never take their edge to $v_4$ and (ii) satisfy $\spec_i$ for every strategy $\pi_{-i}$ of the other player that also never takes their edge to $v_4$, forms an SE strategy profile $(\strat_1,\strat_2)$. These \emph{minimal cooperation obligations} for an SE can be interpreted as a \emph{specification profile} $(\specSE_1,\specSE_2)$, s.t.\ $\specSE_1:=\assump_1\wedge(\assump_2\Rightarrow\spec_1)$ and $\specSE_2:=\assump_2\wedge(\assump_1\Rightarrow\spec_2)$, where $\assump_1 = \LTLalways\neg(v_3\wedge\LTLnext v_4)$ and $\assump_2 = \LTLalways\neg(v_2\wedge\LTLnext v_4)$ express the above discussed assumption that $\p{i}$ does not move to $v_4$ from their vertex.
It turns out, that this new \emph{specification profile} $(\specSE_1,\specSE_2)$ has three nice properties: (i) it is \emph{most general} meaning it does not lose any cooperative solution, i.e., $\spec_1\wedge\spec_2 = \specSE_1\wedge\specSE_2$ , (ii) it is \emph{realizable}, i.e., $\p{i}$ has a strategy $\strati$ that satisfies $\specSEi$ in a zero-sum sense, (i.e., no matter what the other player does) and, most importantly (iii) it is \emph{secure (winning)}, i.e., every strategy profile $(\strat_1,\strat_2)$, where $\p{i}$'s strategy $\strati$ satisfies $\specSEi$ (in a zero-sum sense) is a \emph{winning} SE.
While properties (i) and (iii) motivated us to call the set of new specifications a \emph{most general winning secure equilibrium} (GWSE), property (ii) ensures that any specification $\specSEi$ from this tuple is locally and fully independently realizable by every component. %
Conceptually, this allows us to move from \emph{deviation-punishment} in SE-based implementations to a \emph{distributed, assume-guarantee based realization} of SE.

\smallskip
\noindent\textbf{Contribution.}
By moving from \emph{strategy profiles} (WSE) to \emph{specification profiles} (GWSE) for SE realizations,  our approach takes the conceptualisation of rationality for distributed synthesis to an extreme: as we are in the position to \emph{design} every component (as it is a computer system not a human that actually acts rationally) we can enforce that implementations respect the new specifications $\specSEi$. We only use the \emph{concept of rationality} encoded in WSE to automatically obtain \emph{meaningful and implementable} distributed specifications $\specSEi$ for this co-design process. Thereby the implementation of an accompanying \emph{punishment mechanism} to enforce rationality of players becomes obsolete. 
The obtained flexibility in players' strategy choices can be utilized for robustness and adaptability of local implementations, which makes GWSE particularly suited for embedded systems applications.

Concretely, our contribution is three-fold:
\begin{inparaenum}[(1)]
 \item We formalize GWSE for $k$-player games over finite graphs, where each player has an $\omega$-regular specification.
 \item We devise an \emph{iterative semi-algorithm}\footnote{A semi-algorithm is an algorithm that is not guaranteed to halt on all inputs.} for GWSE synthesis under $\omega$-regular specifications.
 \item We give a (sound but incomplete) \emph{exponential-time algorithm} for GWSE synthesis under \emph{parity} specifications. %
\end{inparaenum}

\smallskip
\noindent\textbf{Other Related Work.}
After the introduction of \emph{secure equilibria} (SE) by Chatterjee et al.~\cite{chatterjee_SE}, there has been several efforts on extending the notion to other classes of games, e.g., games with sup, inf, lim sup, lim inf, and mean-payoff measures~\cite{BruyereMR14}, multi-player games with probabilistic transitions~\cite{PrilFKSV14} or quantitative reachability games~\cite{BrihayeBP14}.
Furthermore, a variant of secure equilibria, called \emph{Doomsday equilibria} was studied in~\cite{ChatterjeeDFR17}, where if any coalition of players deviates and violates one players' objective, then the objective of every player is violated.
Moreover, the notion of secure equilibria has been applied effectively in the synthesis of mutual-exclusion protocols~\cite{ChatterjeeR14,BloemCJK15} and fair-exchange protocols~\cite{KremerR03,LiLXHF14}.

Motivated by similar insights, other concepts of rationality have also been introduced in multi-player games, e.g.\ subgame perfect equilibria~\cite{Ummels06,BriceRB23,Steg22,BruyereRPR21,BriceRB21} or rational synthesis~\cite{FismanKL10,KupfermanS22,FiliotGR18}. 
Similar to the implementations of SE by \cite{chatterjee_SE}, these works restrict implementations to a \emph{single} strategy profile. In contrast, our work introduces a more flexible concept of rationality that is closely related to 
contract-based distributed synthesis, as in~\cite{MajumdarMSZ20,FinkbeinerP22,DammF14,AnandNS23}. 
Here, an assume-guarantee contract is synthesized, such that every strategy realizing the guarantee is ensured to win whenever the other players satisfy the assumption. While this is conceptually similar to our synthesis of GWSE, these works do not consider the players to be adversarial, and hence, there is no notion of \emph{equilibria}.

To the best of our knowledge, the only other work that also combines flexibility with equilibria is
\emph{assume-admissible (AA) synthesis}~\cite{admissible}. Their work utilizes a different, incomparable definition of rationality based on a dominance order. Both approaches are incomparable -- there exist co-synthesis problems where our approach successfully synthesizes a GWSE and no AA contract exists, and vice versa (see \cref{example:algo-admissible} for details). Conceptually, AA contracts still require \emph{rational} behaviour of players within the contract, while our approach only uses rationality as a \emph{concept} to synthesize meaningful local specifications which can then be implemented in an arbitrary (non-rational) manner. We believe that this is a superior strength of our approach compared to AA synthesis.

\section{Preliminaries}\label{section:prelims}

\noindent\textbf{Notation.}
We use $\mathbb{N}$ to denote the set of natural numbers including zero.
Given $a,b\in\mathbb{N}$ with $a<b$, we use $[a;b]$ to denote the set $\set{n\in\mathbb{N} \mid a\leq n\leq b}$.
For any given set $[a;b]$, we write $i\ineven [a;b]$ and $i\inodd [a;b]$ as short hand for $i\in [a;b]\cap \set{0,2,4,\ldots}$ and $i\in [a;b]\cap \set{1,3,5,\ldots}$ respectively.
For a finite alphabet $\Sigma$, $\Sigma^*$ and $\Sigma^\omega$ denote the set of finite and infinite words over $\Sigma$, respectively.

\smallskip
\noindent\textbf{Linear Temporal Logic (LTL).}
Given a finite set $\AP$ of atomic propositions, \emph{linear temporal logic (LTL)} formulas over $\AP$ are defined by the grammar:
\[\spec\coloneqq p\in\AP \mid \spec\vee\spec \mid \neg\spec \mid \LTLnext\spec \mid \spec\LTLuntil\spec,\]
where $\vee$, $\neg$, $\LTLnext$, and $\LTLuntil$ denotes the operators \emph{disjunction}, \emph{negation}, \emph{next}, and \emph{until}, respectively. 
Furthermore, we use the usual derived operators, $\true = p\vee\neg p$, $\false = \neg\true$, \emph{conjunction} $\spec\wedge\spec' = \neg(\neg\spec\vee\neg\spec')$, \emph{implication} $\spec\Rightarrow\spec' = \neg\spec\vee\spec'$, and other temporal operators such as \emph{finally} $\LTLeventually\spec = \true\LTLuntil\spec$ and \emph{globally} $\LTLalways\spec = \neg\LTLeventually\neg\spec$.
The semantics of LTL formulas are defined as usual (see standard textbooks~\cite{baier2008principles}).

\smallskip
\noindent\textbf{Game Graphs.}
A \emph{$k$-player (turn-based) game graph} is a tuple $\gamegraph= \tup{V,E,v_0}$ where $(V,E,v_0)$ is a finite directed graph with \emph{vertices} $ V $ and \emph{edges} $ E $, and $v_0\in V$ is an initial vertex.
For such a game graph, let $\players = [1;k] $ be the set of players such that  $V=\bigcup_{i\in\players} V_i$ is partioned into vertices of $k$ players in $\players$.
We write $E_i$, $i\in\players$, to denote the edges from $\p{i}$'s vertices, i.e., $E_i = E\cap(V_i\times V)$.
Further, we write $V_{-i}$ and $E_{-i}$ to denote the set $\bigcup_{j\neq i}V_j$ and $\bigcup_{j\neq i} E_j$, respectively.
A \emph{play} from a vertex $u_0$ is a finite or infinite sequence of vertices $\play=u_0u_1\ldots$ with $(u_j,u_{j+1})\in E$ for all $j\geq 0$. 

\smallskip
\noindent\textbf{Specifications.}
Given a game graph $\gamegraph$, we consider \emph{specifications} specified using a LTL formula $\Phi$ over the vertex set~$V$, that is, we consider LTL formulas whose atomic propositions are sets of vertices~$V$. 
In this case the set of desired infinite plays is given by the semantics of $\spec$ over $\gamegraph$, which is an $\omega$-regular language $\lang(\gamegraph,\spec)\subseteq V^\omega$. 
We just write $\lang(\spec)$ to denote this language when the game graph $\gamegraph$ is clear in the context.
Every game graph with an arbitrary $\omega$-regular set of desired infinite plays can be reduced to a game graph (possibly with an extended set of vertices) with an LTL objective, as above. 
The standard definitions of $\omega$-regular languages are omitted for brevity and can be found in standard textbooks~\cite{baier2008principles}. 
To simplify notation we use $ e=(u,v) $ in LTL formulas as syntactic sugar for $ u\wedge \bigcirc v $. 

\smallskip
\noindent\textbf{Games and Strategies.}
A \emph{$k$-player game} is a pair $\game=\tup{\gamegraph,(\spec_i)_{i\in\players }}$ where $\gamegraph$ is a $k$-player game graph and 
each $ \spec_i $ is an \emph{objective} for $\p{i}$ over $\gamegraph$.
A strategy of $\p{i},~i\in\players $, is a function $\strati\colon \vertex^*\vertexi\to \vertex$ such that for every $\playprefix v \in \vertex^*\vertexi$, it holds that $(v,\strati(\playprefix v))\in E$.
A strategy profile for a set of players $\players' \subseteq \players $ is a tuple $\Strat = (\strati)_{i\in\players'}$ of strategies, one for each player in $\players'$.
To simplify notation, we write $\playersnoti$ and $\stratnoti$ to denote the set $\players\setminus\{i\}$ and their strategy profile $(\stratj)_{j\in\players\setminus\{i\}}$, respectively.
Given a strategy profile $(\strati)_{i\in\players'}$, we say that a play $\play=u_0u_1\ldots$ is \emph{$(\strati)_{i\in\players'}$-play} if for every $i\in\players'$ and for all $\ell\geq 1$, it holds that $u_{\ell-1}\in \vertexi$ implies $u_{\ell} = \strati(u_0\ldots u_{\ell-1})$.

\smallskip
\noindent\textbf{Satisfying Specifications.}
Given a game graph $\gamegraph$ and a specification $\spec$, a play $\play$ \emph{satisfies} $\spec$ if $\play\in\lang(\spec)$.
A strategy profile $(\strati)_{i\in\players'}$ satisfies/winning w.r.t. a specification~$\spec$, from a vertex $v$, denoted by $(\strati)_{i\in\players'}\vDash_v \spec$, if every $(\strati)_{i\in\players'}$-play from $v$ satifies~$\spec$.
We just write $(\strati)_{i\in\players'}\vDash \spec$ if $v$ is the initial vertex.
We collect all vertices from which there exists a strategy profile for players in $\players'$ that satisfies $\spec$ in the winning region\footnote{Slightly abusing notation, we write $\team{i}\spec$ for singleton sets of players $\players'=\{i\}$.}~$\team{\players'}(\gamegraph,\spec)$. We just write $\team{\players'}\spec$ to denote this set if game graph $\gamegraph$ is clear in the context.
Furthermore, we write $\specnoti$ to denote $\bigwedge_{j\in\playersnoti}\spec_j$.

\smallskip
\noindent\textbf{Parity Specifications.} 
Give a game graph $\gamegraph = (V,E,v_0)$, a specification $\spec$ is called \emph{parity} if $\spec =\textstyle\parity(\priority)\coloneqq \bigwedge_{i\inodd [0;d]} \left(\LTLalways\LTLeventually \priorityset{i}\Rightarrow \bigvee_{j\ineven [i+1;d]} \LTLalways\LTLeventually \priorityset{j}\right)$, 
with $ \priorityset{i}=\{v\in V\mid \priority(v)=i \} $ for some priority function $ \priority: V\rightarrow [0;d] $ that assigns each vertex a priority. A play satisfies such a specification if the maximum of priorities seen infinitely often is even.

\section{Most General Winning Secure Equilibria}\label{section:equilibrium}
This section formalizes most general winning secure equilibria (GWSE). In order to do so, we first recall the notion of secure equilibria from \cite{chatterjee_SE}.

\smallskip
\noindent\textbf{Secure Equilibria.}
Given a \emph{$k$-player game} $\game=\tup{\gamegraph,(\spec_i)_{i\in\players}}$ and a strategy profile $\Strat:=(\strati)_{i\in\players}$ one can define a payoff profile, denoted by $\payoff(\Strat)$, as the tuple $(p_i)_{i\in\players}$ s.t.\ $p_i=1$ iff $\Strat\vDash\spec_i$. With this, we can define a $\p{j}$ preference order $\prec_j$ on payoff profiles lexicographiacally, s.t.\
\begin{equation*}
(p_i)_{i\in\players} \prec_j (p_i')_{i\in\players} \text{ iff } (p_j<p_j')\vee \big((p_j = p_j')\wedge(\forall i\neq j. p_i\geq p_i')\wedge (\exists i\neq j. p_i > p_i')\big).
\end{equation*}
Intuitively, this preference order captures the fact that every player's main objective is to satisfy their own specification $\spec_i$, and, as a secondary objective, falsify the specifications of the other players. 

\begin{definition}
 Given a \emph{$k$-player game} $\game=\tup{\gamegraph,(\spec_i)_{i\in\players }}$, a strategy profile $\Strat:=(\strati)_{i\in\players}$ is a \emph{secure equilibrium} (SE) if for all $i\in\players$, there does not exist a strategy $\strati'$ of $\p{i}$ such that $\payoff(\Strat)\prec_i \payoff(\strati',\stratnoti)$.
\end{definition}

It is well known that every secure equilibrium is also a nash equilibrium in the classical sense. Within this paper, we only consider \emph{winning secure equilibria} (WSE) i.e., SE with the payoff profile $(p_i=1)_{i\in\players}$.
As WSE have a trivial payoff profile, they can be characterized without referring to payoffs as formalized next. %

\begin{definition}\label{def:SE}
    Give a $k$-player game $(\gamegraph,(\spec_i)_{i\in\players })$, a \emph{winning secure equilibrium (WSE)} is a strategy profile $(\strati)_{i\in\players }$ such that
    \begin{inparaenum}[(i)]
        \item $(\strati)_{i\in\players } \vDash \bigwedge_{i\in\players } \speci$; and \label{item:def:SE1}
        \item for every strategy $\strati'$ of $\p{i}$, if 
        $(\strati',\strat_{-i})\not\vDash\spec_{-i}$ holds, then $(\strati',\strat_{-i})\not\vDash\speci$ holds.\label{item:def:SE2}
    \end{inparaenum}   
\end{definition}
Intuitively, \cref{item:def:SE1} ensures that the strategy profile satisfies all player's objective, whereas \cref{item:def:SE2} ensures that no player can improve, i.e., falsify another player's objective without falsifying their own objective, by deviating from the prescribed strategy.

\smallskip
\noindent\textbf{Most General Winning Secure Equilibria.}
As illustrated by the motivating example in \cref{section:intro}, we aim at generalizing WSE from single \emph{strategy profiles} to \emph{specification profiles} that capture an infinite number of WSE. These specification profiles $(\specSEi)_{i\in\players}$, which we call \emph{most general winning secure equilibria} (GWSE), allow each player to locally (and fully independently) pick a strategy~$\strati$ that is winning for $\specSEi$ (in a zero-sum sense). It is then guaranteed that any resulting strategy profile $(\strati)_{i\in\players }$ is indeed a WSE. This is formalized next.

\begin{definition}\label{def:equilibrium}
    Give a $k$-player game $(\gamegraph,(\spec_i)_{i\in\players })$, a tuple $(\specSEi)_{i\in\players }$ of specifications is said to be a \emph{most general winning secure equilibrium (GWSE)} if it is
    \begin{enumerate}[(i)]
        \item (most) general: $\lang(\bigwedge_{i\in\players }\specSEi) = \lang(\bigwedge_{i\in\players }\speci)$;\label{item:def:equilibrium:general}
        \item realizable: $v_0\in\team{i}\specSEi$ for all $i\in\players $; and \label{item:def:equilibrium:realizable}
        \item secure (winning): every strategy profile $(\strati)_{i\in\players }$ with $\strati\vDash\specSEi$ is a WSE.\label{item:def:equilibrium:stable}
    \end{enumerate}
\end{definition}
Intuitively, generality ensures that the transformation of the specifications $(\spec_i)_{i\in\players }$ into new specifications $(\specSEi)_{i\in\players }$ does not lose any winning play. Further, realizability ensures that every single player can enforce $\specSEi$ (without the help of other players) from the initial vertex. Finally, security ensures that any locally chosen strategy $\strati$ winning for $\specSEi$ fors a strategy profile which is indeed a WSE.

\section{Computing GWSE in $\omega$-regular Games}\label{section:GWSEforALL}

This section proposes an \emph{iterative semi-algorithm}\footnote{A semi-algorithm is an algorithm that is not guaranteed to halt on all inputs.} to compute GWSE in this paper which utilizes the concept of \emph{adequately permissive assumptions} (APA) introduced by Anand et al.~\cite{APA}. Given a $k$-player game $(\gamegraph,(\spec_i)_{i\in\players })$, an APA is a specification $\assump_i$ that collects all $\p{i}$ strategies which allow for a cooperative solution if other players cooperate. It therefore overapproximates the set of all $\p{i}$ strategies which could possibly form a WSE with the other players. As a consequence, the intersection $\bigwedge_{i\in\players}\assumpi$ is an overapproximation of a GWSE. In order to refine this approximation, the next computation round can now use the APA's of other players when computing new local APA's. %
In order to properly formalize this idea, we first recall the concept of APA's from~\cite{APA}.

\subsection{Adequately Permissive Assumptions}\label{subsec:APA}
Following~\cite{APA}, we define an \emph{adequately permissive assumption} (APA) as follows.

\begin{definition}\label{def:APA}
    Given a $k$-player game graph $\gamegraph = (V,E,v_0)$ and a specification~$\spec$, we say that  a specification $\assumpi$ is an \emph{adequately permissive assumption (APA)} on $\p{i}$ for $\spec$ if it is:
    \vspace{-0.08cm}
    \begin{enumerate}[(i)]
        \item \emph{sufficient}: there exists a strategy profile $\strat_{-i}$ such that for every $\p{i}$ strategy $\strati$ with $\strati\vDash\assumpi$, we have $(\strat_i,\strat_{-i})\vDash\spec$;
        \item \emph{implementable}: $\team{i}\assumpi = V$; and
        \item \emph{permissive}: $\lang(\assumpi)\supseteq\lang(\spec)$.
    \end{enumerate}
\end{definition}

The intuition behind an APA is that even if a player can not realize a specification $\spec$, they should at least satisfy an APA on them as it will allow them to realize $\spec$ if the other players are willing to help (sufficiency). Further, such a behavior by $\p{i}$ does not prevent any WSE (permissiveness), and $\p{i}$ can individually choose to follow an APA (implementability).

\begin{remark}\label{remark:APA}
 While \cref{def:APA} is an almost direct adaptation from \cite[Def. 2-5]{APA} to $k$-player games, it has a couple of noteable differences. %
First, Anand et al. define APA's for $2$-player games and, conceptually, use APA's to constraint the opponents moves. While we can simply view the $k$-player game as a $2$-player game between the protagonist $\p{i}$ and (the collection of) its opponents $\playersnoti$, we will use the computed assumption $\assump_i$ to constrain the \emph{protagonist's} moves (not the opponent) in \cref{def:APA}.
Second, the sufficiency condition for an APA in \cite[Def. 2]{APA} does not depend on an initial vertex. An APA always exists in their setting (possibly being $\true$ when $\teamall \spec = \emptyset$). In contrast, the $k$-player games in this paper have a designated initial vertex, %
hence, an APA only exists iff $v_0\in\teamall \spec$.
\end{remark}

With this insight, we can use the algorithm from \cite{APA} to compute APA's for parity specificatios $\spec=\parity(\priority)$ in polynomial time.

\begin{lemma}[{\cite[Thm. 4]{APA}}]\label{lemma:computeAPA}
    Given a $k$-player game graph $\gamegraph = (V,E,v_0)$ and a parity specification $\spec=\parity(\priority)$, an APA on $\p{i}$ for $\spec$ can be computed, if one exists, in time $\bigO(\abs{V}^4)$. 
\end{lemma}
Let us write $\computeAPA(\gamegraph, \spec, i)$ to denote the procedure that returns this APA if it exists; otherwise, it returns $\false$.

\begin{remark}\label{rem:nonparityAPA}
We note that $\cref{lemma:computeAPA}$ also gives a method to compute APA's for games with LTL- or $\omega$-regular specifications as such games can be converted into parity games (possibly with an extended game graph) by standard methods~\cite{baier2008principles}. Therefore, with a slight abuse of notation, we will also call the algorithm $\computeAPA(\gamegraph, \spec, i)$ if $\spec$ is not a parity specification, which the understanding, that the game is always converted into a parity game first. This might incur an exponential blowup of the state space. 
As we call $\computeAPA$ repeatedly to compute GWSE's, this blowup might cause non-termination (see Sec.~\ref{subsec:termination-tractability} for details). In order to obtain a (non-optimal but) terminating algorithm for GWSE computation, we will mitigate this blowup later in \cref{section:fastAlgo}.
\end{remark}

\subsection{Iterative Computation of APA's}\label{subsec:itterAPA}
Given the results of the previous section, we can use the algorithm $\computeAPA$ on a given game $(\gamegraph,(\spec_i)_{i\in\players })$ to compute APA's for each player, i.e., \linebreak
$\assumpi:=\computeAPA(\gamegraph, \spec_i, i)$. Intuitively, $\assumpi$ overapproximates the set of all $\p{i}$ strategies which could possibly form a WSE with the other players. As a consequence, the intersection $\bigwedge_{i\in\players}\assumpi$ is an overapproximation of the GWSE.

As outlined previously, we will iteratively refine these computed APA's to finally compute the GWSE. In order to do so, we want to condition the computation of the next-round APA $\assumpi'$ on the previous-round APA's of all other players $\assumpnoti$, as any secure strategy of players in $\playersnoti$ is incentivized to comply with $\assumpnoti$. The most intuitive method to do this is to simply consider $\assumpnoti\Rightarrow\spec_i$ as the specification for APA computation in the next round. However, the way sufficiency is formulated for APA's prevents this approach, as the implication $\assumpnoti\Rightarrow\spec_i$ is true if $\assumpnoti$ is false. As there obviously exists a strategy profile $\stratnoti$ which violates $\assumpnoti$, the sufficiency condition becomes meaningless for this specification. 

However, as we know that $\assumpnoti$ are APA's, their implementability constraint (\cref{def:APA}.ii) ensures that $\p{i}$ can neither enforce nor falsify them. Therefore, a new specification  $\spec_i':=\assumpnoti\wedge\spec_i$ still puts all the burden of satisfying $\assumpnoti$ to players in $\playersnoti$ and hence, implicitly constrains the choices of $\playersnoti$ to strategies complying with $\assumpnoti$ for sufficiency of the new APA. However, using  $\spec_i':=\assumpnoti\wedge\spec_i$ indeed weakens the permissiveness requirement $\lang(\assumpi')\supseteq\lang(\spec\wedge\assump_{-i})$, i.e., the new APA $\assumpi'$ needs to be more general than the specification $\spec$, only when the assumption $\assump_{-i}$ holds.
With these refined conditions for sufficiency and permissiveness, it becomes evident that an APA for specification $\spec$ under assumption $\assump_{-i}$ is equivalent to an APA for the modified specification $\assump_{-i}\wedge\spec$, as formalized below.

\begin{definition}\label{def:APAunderA}
    Given a $k$-player game, a specification $\spec_i$ and an assumption $\assumpnoti$, we say that the specification $\assumpi$ is an \emph{APA on $\p{i}$ for $\spec_i$ under $\assumpnoti$} if it is an APA on $\p{i}$ for specification $\assumpnoti\wedge\spec$.
\end{definition}
Following \cref{rem:nonparityAPA}, we denote by $\computeAPA(\gamegraph,\assump_{-i}\wedge\spec,i)$ the algorithm which computes \emph{APA's} on $\p{i}$ for $\spec$ \emph{under assumptions} $\assump_{-i}$, even though $\assump_{-i}\wedge\spec$ is typically not a parity specification over $\gamegraph$ anymore.

\subsection{Computing GWSE}\label{sec:GWSE}
Using all the intuition discussed before, we now give a semi-algorithm in \cref{algo:computeGE} to compute GWSE for $k$-player games with $\omega$-regular specifications for all players. 
The main idea is to iteratively compute assumptions $(\assumpi)_{i\in\players }$ on every player and check if they are stable enough so that every player can satisfy their actual specification~$\speci$ under the assumption $\assump_{-i}$.
If not, then, in the next iteration, we compute new assumptions $(\assumpi')_{i\in\players }$ that are stricter than earlier ones, i.e., $\lang(\assumpi')\subseteq\lang(\assumpi)$ but still more general than their specifications under the earlier assumption, i.e., $\lang(\assumpi')\supseteq\lang(\assump_{-i}\wedge\speci)$. 

\vspace{-0.5cm}
\begin{algorithm}
	\caption{$\computeGE(\game)$}\label{algo:computeGE}
	\begin{algorithmic}[1]
		\Require A $k$-player game $\game$ with game graph $\gamegraph = (V,E,v_0)$ and parity specifications $(\speci)_{i\in\players }$. 
		\Ensure Either a GWSE $(\specSEi)_{i\in\players }$ or $\false$.
        \State $\assumpi \gets \true$ $\forall i\in\players$\label{algo:computeGE:initial}
        \State \Return $\recursiveGE(\game,(\assumpi)_{i\in\players })$\label{algo:computeGE:rec}

        \Statex
		\Procedure {$\recursiveGE$}{$\game,(\assumpi)_{i\in\players }$}
        \State $\specSEi\gets \assumpi \wedge (\assump_{-i}\Rightarrow \speci)$ $\forall i\in\players$\label{algo:computeGE:specSEi}
        \If {$v_0\in \bigcap_{i\in\players }\team{i}\specSEi$}\label{algo:computeGE:if1}
            \State \Return $(\specSEi)_{i\in\players }$\label{algo:computeGE:terminate1}
        \EndIf
        \State $\assumpi' \gets \assumpi \wedge\computeAPA(\gamegraph,\assump_{-i}\wedge\speci,i)$ $\forall i\in\players$\label{algo:computeGE:assumpi'}
        \If {$\assumpi' = \assumpi$ for all $i\in\players $}\label{algo:computeGE:assumpi_eq}
        \State \Return $\false$
        \EndIf
        \State \Return $\recursiveGE(\game,(\assumpi')_{i\in\players })$
        \EndProcedure
	\end{algorithmic}
\end{algorithm}
\vspace{-0.5cm}

More specifically, we start with $\assumpi=\true$ for each $i\in\players $ in the first iteration (\cref{algo:computeGE:initial}), and then in every iteration, we want each player to satisfy $\specSEi = \assumpi\wedge(\assump_{-i}\Rightarrow\speci)$ (computed in \cref{algo:computeGE:specSEi}) \emph{by themselves}, i.e., always satisfy their assumption~$\assumpi$ and satisfy specification $\speci$ whenever others satisfy their assumptions~$\assump_{-i}$. Note that, in this part of the algorithm it is correct to use this implication-style specification, as it is used for solving a \emph{zero-sum $2$-player game} between $\p{i}$ and its opponent (i.e., the collection of all other players in $\playersnoti$) for the specification $\specSEi$. The winning regions $\team{i}\specSEi$ for each such zero-sum $2$-player game are then intersected in \cref{algo:computeGE:if1} to obtain the winning region that is achievable by any strategy profile $(\strati)_{i\in\players}$ where $\strati$ is a winning strategy of $\p{i}$ w.r.t.\ $\specSEi$ (in a zero-sum sense). If this resulting winning region contains the initial vertex, we return the specification $(\specSEi)_{i\in\players }$ (\cref{algo:computeGE:terminate1}), which is proven to indeed be a GWSE in \cref{thm:computeGE}.

If this is not the case, we keep on strengthening APA's, as discussed in \cref{subsec:itterAPA}, to make the above mentioned zero-sum 2-player games easier to solve (as they can rely on tighter assumptions now). Hence, we call $\computeAPA$ with the modified specifications $\speci':=\assump_{-i}\wedge\speci$ for all players (\cref{algo:computeGE:assumpi'}). If this assumption refinement step was unsuccessful, i.e., assumptions have not changed (\cref{algo:computeGE:assumpi_eq}), we give up and return $\false$. Otherwise, we recheck the termination condition for the newly computed APA's.

\begin{example}\label{example:algo-admissible}
\begin{figure}[b]
    \centering
    \begin{tikzpicture}
        \node[player2] (v0) at (0, 0) {$v_0$};
        \node[player1] (v1) at (\hpos, 0) {$v_1$};
        \node[player1] (v2) at (\hpos, 0.7*\ypos) {$v_2$};
        \node[player1] (v5) at (2*\hpos, 0) {$v_5$};
        \node[player1] (v3) at (-\hpos, 0) {$v_3$};
        \node[player2] (v4) at (-2*\hpos, 0) {$v_4$};
        
        \path[->] (0,-0.65) edge (v0.south); 
        \path[->] (v0) edge[bend left=20] (v1) edge[bend left=20] (v3) edge[loop above] ();
        \path[->] (v1) edge[bend left=20] (v0) edge (v2) edge (v5);
        \path[->] (v2) edge[loop right] ();
        \path[->] (v3) edge[bend left=20] (v0) edge (v4);
        \path[->] (v4) edge[loop above] ();
        \path[->] (v5) edge[loop above] ();
    \end{tikzpicture}
    \vspace{-0.2cm}
    \caption{A two-player game with initial vertex $v_0$, $\p{1}$'s vertices (squares), $\p{2}$'s vertices (circles) and specifications $\spec_1 = \LTLeventually\LTLalways \{v_5\}$ and $\spec_2 = \LTLeventually\LTLalways \{v_4,v_5\}$.}
    \label{fig:algo-admissible}
\end{figure}
Before proving the correctness of the (semi) \cref{algo:computeGE}, let us first illustrate the steps using an example depicted in \cref{fig:algo-admissible}.
In \cref{algo:computeGE:initial}, we begin with $\assump_1 = \assump_2 = \true$ and run the recursive procedure $\recursiveGE$ in \cref{algo:computeGE:rec}.

Within the first iteration of $\recursiveGE$, in \cref{algo:computeGE:specSEi}, we set $\specSEi = \speci$ as $\assumpi = \true$ for all $i\in[1;2]$.
Then, in \cref{algo:computeGE:if1}, we check whether each player can satisfy $\specSEi=\speci$ without cooperation (i.e., in a zero-sum sense), from the initial vertex $v_0$.
As no player can ensure that, we move to \cref{algo:computeGE:assumpi'}.
Here, as $\assumpi = \true$ for $i\in[1;2]$, the new assumptions $\assumpi'$ is an APA computed by $\computeAPA(\gamegraph,\speci,i)$.
This gives us $\assump_1' = \LTLalways\neg (e_{12}\wedge e_{34})\wedge\LTLeventually\LTLalways\neg e_{10}$ and $\assump_2' = \LTLeventually\LTLalways\neg e_{00}$, where $e_{ij} = v_i\wedge\LTLnext v_j$.
Intuitively, $\assump_1'$ ensures that edges, i.e., $v_1\rightarrow v_2$ and $v_3\rightarrow v_4$, leading to the region from which it is not possible to satisfy $\spec_1$ are never taken; and the edge, i.e., $v_1\rightarrow v_0$, restricting the play to progress towards target vertex $v_5$ (as in $\spec_1$) is eventually not taken. Similarly, $\assump_2$ ensures that the edge $v_0\rightarrow v_0$ is eventually not taken that ensures progress towards $\spec_2$'s target vertices $\{v_4,v_5\}$.
As $\assumpi'\neq\assumpi$ for all $i\in[1;2]$ in \cref{algo:computeGE:assumpi_eq}, we go to the next iteration of $\recursiveGE$. %

In the second iteration, we again compute the new potential GWSE $(\specSE_1,\specSE_2)$ with $\specSEi = \assumpi\wedge(\assumpnoti\Rightarrow\speci)$ in \cref{algo:computeGE:specSEi}. 
In \cref{algo:computeGE:if1}, we find that $v_0\not\in\team{1}\specSE_1$. That is because $\p{1}$ cannot ensure satisfying $\spec_1$ even when $\p{2}$ satisfies $\assump_2$ as $\p{2}$ can always use edge $v_0\rightarrow v_3$ leading to the play $(v_0v_3)^\omega\not\vDash\spec_2$.
Hence, in \cref{algo:computeGE:assumpi'}, the APA under $\assump_1$ gives a more restricted assumptions on $\p{2}$: $\assump_2' = \LTLeventually\LTLalways\neg(e_{00}\wedge e_{03})$.
As the assumption $\assump_2$ on $\p{2}$ was very weak, the APA for $\p{1}$ under $\assump_2$ results in the same assumption as $\assump_1$, and hence, $\assump_1' = \assump_1$. Then, we move to the third iteration.

In this iteration, we find that both players can indeed satisfy their new specification $\specSEi$ from the initial vertex in \cref{algo:computeGE:if1}. Hence, we finally return a GWSE $(\specSE_1,\specSE_2)$ with $\specSE_i = \assumpi\wedge (\assumpnoti\Rightarrow\speci)$ where $\assump_2 = \LTLeventually\LTLalways\neg(e_{00}\wedge e_{03})$ and $\assump_1 = \LTLalways\neg(e_{12}\wedge e_{34})\wedge\LTLeventually\LTLalways\neg(e_{10})$.
\end{example}

\begin{remark}\label{rem:AAsynth}
    Let us remark that for the game depicted in \cref{fig:algo-admissible}, assume-admissible (AA) synthesis~\cite{admissible} has no solution. AA-synthesis utilizes a different, incomparable definition of rationality based on a dominance order. In their framework, a $\p{i}$ strategy $\strati$ is said to be \emph{dominated} by $\strati'$ if the set of strategy profiles that $\strat'$ is winning against (i.e., satisfies $\p{i}$'s specification) is strictly larger than that of $\strat$. A strategy not dominated by any other strategy is called \emph{admissible}.     
    In AA-synthesis, one needs to find an admissible strategy $\strati$ for $\p{i}$ such that for every admissible strategy $\stratnoti'$ for the other player, $(\strati,\stratnoti')\vDash\speci$. In this example, $\p{1}$ has only one admissible strategy $\strat_1$ that always uses $v_1\rightarrow v_5$ and $v_3\rightarrow v_0$. However, with the admissible strategy $\strat_2'$ of $\p{2}$ that always uses $v_0\rightarrow v_3$, we have $(\strat_1,\strat_2')\not\vDash\spec_1$.
\end{remark}

The next theorem shows that \cref{algo:computeGE} is indeed sound.

\begin{restatable}{theorem}{restatecomputeGE}\label{thm:computeGE}
    Let $\game$ be a $k$-player game with game graph $\gamegraph = (V,E,v_0)$ and parity specifications $(\speci)_{i\in\players }$ such that $(\specSEi^*)_{i\in\players } = \computeGE(\game)$, then $(\specSEi^*)_{i\in\players }$ is a GWSE for $\game$.
\end{restatable}
\begin{proof}
    First, observe that $\computeGE$ did not return $\false$ by the premise of the theorem. So, if $\computeAPA$ returned $\false$ in \cref{algo:computeGE:assumpi'}, i.e., $\assumpi'=\false$ for some $i\in\players$, in some $n$-th iteration, then in the $n+1$-th iteration, we have $\assumpi = \false$ and $\assump_{-j} = \false$ for all $j\in\playersnoti$. So, it holds that $v_0\not\in\team{i}\specSEi = \team{i}\false = \emptyset$ and hence, it does not return in \cref{algo:computeGE:terminate1}.
    Furthermore, as $\assump_{-j}\wedge\specj = \false$ for all $j\in\playersnoti$, by sufficiency, $\computeAPA$ returns $\false$ for all $j\in\playersnoti$. Hence, $\assump_j' = \false$ for all $j\in\players$. This would imply (by similar arguments), in $(n+2)$-th iteration, $\assump_j' = \assump_j = \false$ for all $j\in\players $ and hence, the algorithm would return $\false$. Therefore, we can assume $\computeAPA$ never returned $\false$ in any iteration.

    Now, let us claim that in every iteration of $\recursiveGE$, for all $i\in\players$:
    \begin{align*}
        \text{(claim 1)}\quad\lang(\assumpi)\textstyle\supseteq\lang(\bigwedge_{j\in\players }\spec_j),~\text{and}%
        \quad&&\quad
        \text{(claim 2)}\quad\lang(\assumpi)\supseteq\lang(\assump_{-i}\wedge\speci).%
    \end{align*}
    We will prove the claim using induction on the number of itereative calls to $\recursiveGE$. 
    For the base case, observe $\assumpi = \true$ for all $i\in\players $, hence, the claim holds trivially.
    For the induction step, assume that
    claim 1+2 hold in the $n$-th iteration. Then, for all $i\in\players $, as $\assumpi'$ (computed in \cref{algo:computeGE:assumpi'}) is $\assump$ in the next iteration, it suffices to show that $\lang(\assumpi')\supseteq\lang(\bigwedge_{j\in\players }\spec_j)$ and $\lang(\assumpi')\supseteq\lang(\assump'_{-i}\wedge\speci)$. %
    
    By permissiveness of APA (as in \cref{def:APA}), for all $i\in\players $, we have \linebreak
    $\lang(\computeAPA(\gamegraph,\assump_{-i}\wedge\speci,i))\supseteq \lang(\assump_{-i})\cap\lang(\speci)$.
    Hence, by \cref{algo:computeGE:assumpi'}, for all $i\in\players $, we have $\lang(\assumpi') \supseteq \lang(\assumpi)\cap \lang(\assump_{-i})\cap\lang(\speci) = \left(\bigcap_{j\in\players }\lang(\assump_j)\right)\cap\lang(\speci)$, and hence, by claim 1, $\lang(\assumpi')\supseteq \lang(\bigwedge_{j\in\players }\spec_j)$.

    Similarly, for all $i\in\players $, as $\lang(\assumpi') \supseteq \lang(\assumpi)\cap \lang(\assump_{-i})\cap\lang(\speci)$, by claim 2, we also have $\lang(\assumpi')\supseteq\lang(\assump_{-i})\cap\lang(\speci)$.
    Furthermore, by \cref{algo:computeGE:assumpi'}, for all $j\in\players $, we have $\lang(\assump_j)\supseteq\lang(\assump_j')$, and hence, $\lang(\assump_{-i}) = \bigcap_{j\neq i}\lang(\assump_j) \supseteq \bigcap_{j\neq i}\lang(\assump'_j) = \lang(\assump'_{-i})$.
    Therefore, for all $i\in\players $, we have $\lang(\assumpi')\supseteq\lang(\assump'_{-i})\cap\lang(\speci) = \lang(\assump'_{-i}\wedge\speci)$.
    
    Now, we show that \cref{def:equilibrium} (i)-(ii) indeed holds for the tuple $(\specSEi^*)_{i\in\players }$. %
    
    \noindent\textit{(i) (general)} By construction, $\specSEi^* = \assumpi \wedge (\assump_{-i}\Rightarrow \speci)$ for the specifications $(\assumpi)_{i\in\players }$ computed in last iteration.
    Hence, it holds that 
    \begin{align*}
        &\lang\left(\bigwedge_{i\in\players }\specSEi^*\right) = \bigcap_{i\in\players }\lang(\assumpi \wedge (\assump_{-i}\Rightarrow \speci))
        = \bigcap_{i\in\players }\lang(\assumpi) \cap \bigcap_{i\in\players }\lang(\assump_{-i}\Rightarrow \speci)\\
        \quad&= \bigcap_{i\in\players }\lang(\assump_{-i}) \cap \bigcap_{i\in\players }\lang(\assump_{-i}\Rightarrow \speci)
        = \bigcap_{i\in\players }\lang(\assumpnoti\wedge(\assumpnoti\Rightarrow\speci))
        \subseteq \bigcap_{i\in\players }\lang(\speci) = \lang\left(\bigwedge_{i\in\players }\speci\right).
    \end{align*}
    For the other direction, it holds that 
    \begin{align}
        \lang(\specSEi^*) = \lang(\assumpi)\wedge \lang(\assump_{-i}\Rightarrow\speci)
        \supseteq \lang(\assumpi)\cap\lang(\speci)\label{eq:proof:computeGE-general}
    \end{align}
    Then, by claim 1, for all $i\in\players$, we have $\lang(\specSEi^*)\supseteq \lang\left(\bigwedge_{i\in\players }\speci\right)$, and hence, $\lang\left(\bigwedge_{i\in\players }\specSEi^*\right)\supseteq\lang\left(\bigwedge_{i\in\players }\speci\right)$.
    Therefore, $(\specSEi^*)_{i\in\players }$ is general.

    \noindent\textit{(ii) (realizable)} Holds trivially by \cref{algo:computeGE:if1}.

    \noindent\textit{(iii) (secure)} Let $(\strati)_{i\in\players }$ be a strategy profile with $\strati\vDash\specSEi^*$.
    Then, every $(\strati)_{i\in\players }$-play from $v_0$ satisfies $\specSEi^*$ for all $i\in\players $, and hence, $(\strati)_{i\in\players }\vDash\bigwedge_{i\in\players }\specSEi^*$.
    So, by generality, we have $(\strati)_{i\in\players }\vDash \bigwedge_{i\in\players }\speci$.

    Now, to prove \cref{item:def:SE2} of \cref{def:SE}, let $\strati'$ be a strategy of $\p{i}$, and let $\play$ be the $(\strati',\strat_{-i})$-play from $v_0$.
    As before, for all $j\in\players $, we have $\specSE_j^* = \assump_j \wedge (\assump_{-j}\Rightarrow \spec_j)$. 
    So, for every $j\neq i$, $\play\in\lang(\specSE^*_j)\subseteq \lang(\assump_j)$.
    Hence, we have $\play \in \bigcap_{j\neq i}\lang(\assump_j) = \lang(\assump_{-i})$.

    Now, if $\play\in\lang(\speci)$, then $\play\in\lang(\assump_{-i}\wedge\speci)$. Then, by \cref{eq:proof:computeGE-general} and claim 2, we have $\play\in\lang(\specSEi^*)$.
    Furthermore, as $\stratnoti\vDash\specSEnoti^*$, we have $\play\in\lang(\specSEnoti^*)$.
    Therefore, $\play\in\lang(\specSEi^*\wedge\specSE^*_{-i})$, and by generality, $\play\in\lang(\speci\wedge\spec_{-i})\subseteq \lang(\spec_{-i})$.
    Then, by contraposition, \cref{item:def:SE2} of \cref{def:SE} holds for $(\strati)_{i\in\players }$.
    Hence, $(\strati)_{i\in\players }$ is an SE, and hence, $(\specSEi^*)_{i\in\players }$ is secure.
\qed
\end{proof}

\subsection{Games with an Environment Player}\label{subsec:gamesEnv}
Up to this point, we have only considered games played between $k$ players, each representing a distinct system.
However, in the context of reactive synthesis problems, a different setup is often encountered. Here, the system players play against an environment player, who is considered as being adversarial toward all the system players. Consequently, the system players must fulfill their objectives against all possible strategies employed by the environment player.

Interestingly, this framework can be seen as equivalent to a $(k+1)$-player game with the original $k$ system players and a $(k+1)$-th player, representing the environment. For this new player, the objective is simply $\spec_{k+1} = \true$.
Then, it is easy to see that an APA for such specification $\spec_{k+1}$ under any assumption is $\true$.
Hence, in each iteration of $\recursiveGE$ in \cref{algo:computeGE}, the associated assumption $\assump_{k+1}$ is also $\true$, and thus, $\specSE_{k+1} = \true\wedge((\bigwedge_{i\in[1;k]}\assumpi) \Rightarrow\true) \equiv \true$.
Consequently, if $\computeGE$ yields a GWSE $(\specSEi^*)_{i\in[1;k+1]}$, the new objective of the environment player, $\specSE_{k+1}^* = \true$, doesn't impose any constraints on the environment's actions.
Therefore, the tuple $(\specSEi^*)_{i\in[1;k]}$ remains secure (as in \cref{def:equilibrium}) for the $k$ system players because the environment player can never violate its new specification $\specSE_{k+1}$. In sum, games featuring an environment player can be effectively handled as a special case, as formally summarized below:
\begin{corollary}\label{corollary:gamesEnv}
    Let $\gamegraph = (V,E)$ be a game graph with $k$ system players, i.e., $\players = [1;k]$, and an environment player $\env$ such that $V = \left(\bigcup_{i\in\players} V_i\right) \uplus V_{\env}$.
    Let $(\speci)_{i\in \players}$ be the tuple of specifications, one for each system player.
    Then, a tuple $(\specSEi)_{i\in\players}$ is a GWSE for $(\gamegraph,(\speci)_{i\in\players})$ if and only if $(\specSEi)_{i\in[1;k+1]}$ with $\specSE_{k+1} = \true$ is a GWSE for the $k+1$-player game $(\gamegraph,(\speci)_{i\in[1;k+1]})$ with $\spec_{k+1} = \true$.
\end{corollary}

Furthermore, in synthesis problems, the choices of the environment are sometimes restricted based on a certain assumption $\spec_\env$. In such scenarios, a viable approach involves updating each system player's specification $\speci$ to $\spec_\env\Rightarrow\speci$ and subsequently utilizing \cref{corollary:gamesEnv} to compute a GWSE.
An alternative approach is to consider a $(k+1)$-player game with specification $\spec_{k+1} = \spec_\env$ for the $(k+1)$-th player.
With this approach, the solution becomes more meaningful, as any strategy profile for the system players satisfying the resulting GWSE allows the environment to satisfy its own assumptions $\spec_\env$. 
This approach nicely complements existing works~\cite{ChatterjeeHL10,MajumdarPS19} that aim to synthesize strategies for systems while allowing the environment to fulfill its own requirement.

\subsection{Partially Winning GWSE}\label{subsec:gamesNonMax}
In the preceding sections, we have presented a method for computing \emph{winning} SE, i.e., equilibria where all players satisfy their objectives. %
However, it's worth noting that in certain scenarios, WSE might not exist (see e.g.\ \cite{chatterjee_SE} for a detailed discussion).
In such cases, a subset $\players'$ of players can still form a coalition, which serves their interests by enabling them to compute a GWSE for their coalition only, while treating the remaining players in $\players\setminus\players'$ as part of the environment. This can be accomplished by computing a GWSE with updated specifications denoted as $(\speci')_{i\in\players}$, wherein $\speci' = \speci$ for all $i\in\players'$ and $\speci' = \true$ for all $i\not\in\players'$. This scenario aligns with the concept of considering an environment from \cref{subsec:gamesEnv}.

It is important to emphasize that for instances where no WSE exists, there might not even exist a \emph{unique maximal} outcome for which an SE is feasible, see~\cite[Sec.~5]{chatterjee_SE} for a simple example. 
As a result, there may be multiple coalitions that can offer different advantages to individual players from the initial vertex.
This scenario presents an intriguing, unexplored challenge for future research.

\subsection{Computational Tractability and Termination}\label{subsec:termination-tractability}
While \cref{algo:computeGE} has multiple desirable properties, additionally supported by the possible extensions discussed in \cref{subsec:gamesEnv,subsec:gamesNonMax}, its computational tractability and termination is questionable for the full class of $\omega$-regular games.

As pointed out in \cref{rem:nonparityAPA}, the application of $\computeAPA$ might require changing the game graph for if the input is not a parity specification. While the \emph{language} of the computed APA is guarantee to shrink in every iteration (see the proof of \cref{thm:computeGE}), this does not guarantee termination of \cref{algo:computeGE} as such a language still contains an infinite number of words. Due to the possibly repeated changes in the game graph for APA comutation, the finiteness of the underlying model can also not be used as a termation argument.

In addition, the need to change game graphs induces a severe computational burden. While this might be not so obvious for the polynomial time algorithm $\computeAPA$, this is actually also the case for the (zero-sum) game solver that needs to be invoked \cref{algo:computeGE:if1} of \cref{algo:computeGE}. As the specification for these games also keeps changing in each iteration, a new parity game needs to be constructed in each iteration, which might be increasingly harder to solve, depending on the nature of the added assumptions. 
We will see in \cref{section:fastAlgo} how these problems can be resolved by a suitable restriction of the considered assumption class.

\section{Optimized Computation of GWSE in Parity Games}\label{section:fastAlgo}
As discussed in \cref{subsec:termination-tractability}, the potential need to repeatedly change game graphs in the computations of \cref{algo:computeGE:if1,algo:computeGE:assumpi'} in \cref{algo:computeGE} might incur increasing computational costs and prevents a termination guarantee. To circumvent these problems, this section proposes a different algorithm for GWSE synthesis which overapproximates APA's by a simpler assumption class, called UCA's. The resulting algorithm is computationally more tractable and ensured to terminate.
Nevertheless, unlike the semi-algorithm discussed in the previous section, this algorithm may not be able to compute a GWSE in all scenarios where the semi-algorithm can.

\subsection{From APA's to UCA's}
One of the main features of APA's on $\p{i}$ computed by $\computeAPA$ from \cite{APA}, is the fact that they can be expressed by well structured templates using $\p{i}$'s edges, 
namely \emph{unsafe-edge-}, \emph{colive-edge-}, and \emph{(conditional)-live-group-templates}. Unsafe- and colive-edge-templates are structurally very simple. Given a set of unsafe edges $S\subseteq E_i$ and colive edges $C\subseteq E_i$ the respective assumption templates $\assumpsafe(S):=\bigwedge_{e\in\safegroup}\LTLalways\neg e$ and $\assumpcolive(C):=\bigwedge_{e\in\colivegroup}\LTLeventually\LTLalways\neg e$ simply assert that unsafe (resp. colive) edges should never (resp. only finitely often) be taken. 
We call an assumption which can be expressed by these two types of templates an \textbf{U}nsafe- and \textbf{C}olive-edge-template \textbf{A}ssumption (UCA), as defined next.
\begin{definition}
 Given a $k$-player game graph $\gamegraph = (V,E)$, a specification $\assump$ is called an unsafe- and colive-edge-template assumption (UCA) for $\p{i}$, if there exist sets $\safegroup,\colivegroup\subseteq E_i$ s.t.\ $\assump:=\assumpsafe(S) \wedge \assumpcolive(C)$. We write $\assumptemp{\safegroup}{\colivegroup}$ to denote such assumptions.
\end{definition}

It was recently shown by Schmuck et al.~\cite{vmcai} that two-player (zero-sum) parity games under UCA assumptions, i.e., games $(\gamegraph,\assump\Rightarrow\spec)$ where $\assump$ is an UCA and $\spec$ is a parity specification over $\gamegraph$, can be directly solved over $\gamegraph$ without computational overhead, compared to the non-augmented version $(\gamegraph,\spec)$ of the same game. Interestingly, the synthesis problem under assumptions becomes proveably harder if live-group-templates $\assumpgrlive$ are needed to express an assumption, requiring a change of the game graph in most cases. Conditional-live-group-templates $\assumpgrlive$, are structurally more challenging than UCA's, as they impose a Streett-type fairness conditions on edges in $\gamegraph$ (see \cite[Sec.4]{APA} for details).

Motivated by this result, we will restrict the assumption class used for GWSE computation to UCA's in this section. Unfortunately, UCA's are typically not expressive enough to capture APA's for parity games. This follows from one of the 
main results of Anand et al., which shows that APA's computed by $\computeAPA$ for parity games are expressible by a conjunctions of all \emph{three} template types, as re-stated in the following proposition.

\begin{proposition}[{\cite[Thm. 3]{APA}}]
 Given the premisses of \cref{lemma:computeAPA}, the APA computed by $\computeAPA$ on $\p{i}$ can be written as the conjunction $\assump:=\assumpsafe(S) \wedge \assumpcolive(C) \wedge \assumpgrlive$
 where $\safegroup,\colivegroup\subseteq E_i$.
 \end{proposition}

We therefore need to overapproximate APA's by UCA's, by simply dropping the $\assumpgrlive$-term from their defining conjunction, as formalized next.

\begin{definition}
 Given the premisses of \cref{lemma:computeAPA}, let $\assump:=\computeAPA(\gamegraph,\spec,i)=\assumpsafe(S) \wedge \assumpcolive(C) \wedge \assumpgrlive$.
Then we denote by $\approxAPA(\gamegraph,\spec,i)$ the algorithm that computes $\assumptemp{\safegroup}{\colivegroup}$ by first executing $\computeAPA(\gamegraph,\spec,i)$ and then dropping all $\assumpgrlive$-terms from the resulting APA.
\end{definition}

It is easy to see that $\lang(\assump)\subseteq\lang(\assumptemp{\safegroup}{\colivegroup})$. Therefore, it also follows that $\assumptemp{\safegroup}{\colivegroup}$ is implementable and permissive (i.e., \cref{def:APA}(ii) and (iii) holds). Unfortunatly, $\assumptemp{\safegroup}{\colivegroup}$ is in general no longer sufficient (i.e., \cref{def:APA}(i) does not necessarily hold). 
As the proof of \cref{thm:computeGE} only uses permissiveness of APA, even though sufficiency is lost for UCA's, replacing $\computeAPA$ by $\approxAPA$ in \cref{algo:computeGE} does not mitigate soundness, i.e., whenever $\computeGE$ terminates in \cref{algo:computeGE:terminate1} with a specification profile $(\specSE_i)_{i\in\players}$, this profile is indeed a GWSE, even if APA's are over-approximated by UCA's. This is formalized next.

\begin{theorem}\label{thm:approxAPA}
Let $A\computeGE$ be the algorithm obtained by replacing procedure $\computeAPA$ by $\approxAPA$ in \cref{algo:computeGE}. Then, given a $k$-player game $\game$ with parity specifications such that $(\specSEi^*)_{i\in\players } = A\computeGE(\game)$, the tuple $(\specSEi^*)_{i\in\players }$ is a GWSE for $\game$.
\end{theorem}

The rest of this section will now show how the restriction to UCA's allows to execute \cref{algo:computeGE:if1,algo:computeGE:assumpi'} in \cref{algo:computeGE} efficiently and allows to prove termination of the resulting algorithm for GWSE computation.

\subsection{Iterative Computation of UCA's}\label{sec:computeUCA}
We have seen in the previous section that UCA's can be computed by utilizing $\computeAPA$ and dropping all  $\assumpgrlive$ terms (called $\approxAPA$). Of course, this can be done in every iteration of $\computeGE$. However, $\computeAPA$ expects a party game as an input, and from the second iteration of $\computeGE$ onward the input to $\computeAPA$ is given by $(\gamegraph,\assumpnoti\wedge\speci,i)$, where $\assumpnoti$ is an assumption on players in $\playersnoti$, which is not necessarily a parity game. 

This section therefore provides a new algorithm, called $\computeUCA$ and given in \cref{algo:computeUCA} which computes UCA's for $\p{i}$ directly on the game graph $\gamegraph$ for games $(\gamegraph,\assump\wedge\spec)$ where $\assump = \assumptemp{\safegroup}{\colivegroup}$ is an UCA for $\playersnoti$ with unsafe edges $S\subseteq E_{-i}$ and colive edges $C\subseteq E_{-i}$, and $\spec$ is a parity specification, both over $\gamegraph$. Intuitively, $\computeUCA$ first slightly modifies $\gamegraph$ to a new two-player game graph $\hat{\gamegraph}$ (\cref{algo:computeUCA:graph1,algo:computeUCA:graph2}) s.t.\ the specification $\assump\wedge\spec$ can be directly expressed as a parity specification $\hat{\spec}$ on $\hat{\gamegraph}$ (\cref{algo:computeUCA:graph3}). 
This allows to apply $\approxAPA$ to construct and return an UCA for $\p{1}$ on $\hat{\gamegraph}$ (\cref{algo:computeUCA:compute}).
As the resulting UCA is for $\p{i}$, the unsafe edge and colive edge sets are subsets of $E_i$.
Further, due to the mild modifications from $\gamegraph$ to $\hat{\gamegraph}$, the edges of $\p{i}$ are retained in $\hat{\gamegraph}$ as $E_1$, hence, the resulting UCA is a well-defined UCA for $\p{i}$ in $\gamegraph$.

We have the following soundness result for showing equivalence between the UCA's computed by $\computeUCA$ and $\approxAPA$ for UCA assumptions, proven in \cref{appendix:UCA}.

\begin{algorithm}
	\caption{$\computeUCA(\gamegraph,\assumptemp{\safegroup}{\colivegroup}\wedge\spec,i)$}\label{algo:computeUCA}
	\begin{algorithmic}[1]
		\Require A $k$-player game graph $\gamegraph = (V,E,v_0)$ and specification $\assump\wedge\spec$ with UCA $\assump = \assumptemp{\safegroup}{\colivegroup}$ for $\playersnoti$, i.e., $S,C\subseteq E_{-i}$, and $\spec=\parity(\priority)$ s.t.\ $\priority:V\rightarrow [0;2d+1]$.
		\Ensure An UCA $\assumptemp{\safegroup'}{\colivegroup'}$ for $\p{i}$.
        \State $\hat{V}_1 \gets V_i$ and $\hat{V}_2 \gets V_{-i}\uplus \colivegroup$\label{algo:computeUCA:graph1}
        \State $\hat{E_1} \gets E_i$ and $\hat{E_2} \gets E_{-i}\setminus(\safegroup\cup\colivegroup) \cup \{(u,c),(c,v) \mid c = (u,v)\in\colivegroup\}$\label{algo:computeUCA:graph2}
        \State $\hat{\priority} = \begin{cases}
        \priority(v) &\text{if $v\in V$}\\
        2d+1 &\text{otherwise.} 
    \end{cases}$\label{algo:computeUCA:parity}
        \State $\hat{\gamegraph} = (\hat{V}_1\uplus \hat{V}_2,\hat{E_1}\uplus\hat{E_2},v_0)$; $\hat{\spec}\gets\parity(\hat{\priority})$\label{algo:computeUCA:graph3}
        \State \Return $\approxAPA(\hat{\gamegraph},\hat{\spec},1)$\label{algo:computeUCA:compute}
	\end{algorithmic}
\end{algorithm}

\begin{restatable}{proposition}{restateUCA}\label{prop:computeUCA}
    Given game graph $\gamegraph = (V,E,v_0)$ with parity specification $\spec$ and an UCA $\assump = \assumptemp{\safegroup}{\colivegroup}$ for $\playersnoti$, 
    let $\assump':=\approxAPA(\gamegraph,\assump\wedge\spec,i)$ and 
    $\assump'':=\computeUCA(\gamegraph,\assump\wedge\spec,i)$ then $\lang(\assump')=\lang(\assump'')$.
    Furthermore, $\computeUCA$ terminates in time $\bigO((\abs{V}+\abs{E})^4)$.
\end{restatable}
The proof of this result is given in appendix, and essentially relies on the observation that the parity specification $\hat{\spec}$ in $\hat{\gamegraph}$ expresses the language $\lang(\assump\wedge\spec)$ when restricted to $V$, i.e, $\lang(\hat{\gamegraph},\hat{\spec})|_V = \lang(\gamegraph,\spec\wedge\assump)$ and the fact that every UCA for $\p{1}$ in $\hat{\gamegraph}$ is also an UCA for $\p{i}$ in $\gamegraph$ as proven in \cref{lemma:graphUnsafeColive} in \cref{appendix:UCA}.

The usefulness of expressing the computed assumptions as unsafe and colive edge sets $S,C$ over the input game graph $\gamegraph$ is that there are only a finite number of edges in that graph. Therefore, there obviously also exists only a finite number of unsafe or colive edge sets, which could all be enumerated in the worst case. Therefore, computing UCA's on the same game graph in every iteration, will ensure termination of the overall computation of GWSE.

\subsection{Solving Parity Games under UCA's}\label{sec:computeWin}
As the final step towards an optimized version of \cref{algo:computeGE}, we now address the computations required in \cref{algo:computeGE:if1} of \cref{algo:computeGE}. Observe that this line requires to check $v_0\in \bigcap_{i\in\players }\team{i}\specSEi$ for $\specSEi=\assumpi \wedge (\assump_{-i}\Rightarrow \speci)$. If this check returns $\true$ the algorithm terminates, if it returns $\false$ new assumptions are computed. In both cases, the game graph used to check this conditional will not have any effect on the future behavior of the algorithm. %

Nevertheless, we utilize the recent result by Schmuck et al.\ \cite{vmcai} to compute $\team{i}\specSEi$ more efficiently if $\assumpi$ and $\assumpnoti$ are UCA's on $\p{i}$ and $\playersnoti$, respectively. The construction uses 
the same idea as presented in \cref{algo:computeUCA} to encode UCA's into a new, slightly modified two-player parity game $(\hat{\gamegraph},\hat{\spec})$ which can then be solved by a standard parity solver, such as Zielonka's algorithm~\cite{Zielonka98}, which return the winning region $\win$ of $\p{1}$ in this new game that corresponds to the winning region of $\p{i}$ in $\gamegraph$. The resulting algorithm is called $\computeWin$ given in \cref{algo:computeWin} in \cref{sec:app:win} and has the property that $v_0\in\team{i}(\gamegraph,\specSE)$ if and only if $v_0\in\win$. This is formalized and proven in \cref{prop:computeWin}, \cref{sec:app:win}.

\subsection{Computation of GWSE via UCA's}\label{sec:ocomputeGE}
With the previously discussed algorithms in place, we are now in the position to propose an optimized, surely terminating algorithm to compute GWSE, called $\ocomputeGE$. Within $\computeGE$ the recursive procedure $\recursiveGE$ is replaced by one which uses the algorithms $\computeUCA$ and $\computeWin$ for UCA's from \cref{sec:computeUCA,sec:computeWin}, as follows

	\begin{algorithmic}[1]
	\footnotesize
		\Procedure {$\recursiveGE$}{$\game,(\assumpi)_{i\in\players }$}
        \State $\specSEi\gets \assumpi \wedge (\assump_{-i}\Rightarrow \speci)$ $\forall i\in\players$
        \State \textcolor{blue}{$\win_i\gets\computeWin(\mathcal{G},\specSEi,i)$}
        \If {$v_0\in \bigcap_{i\in\players }\textcolor{blue}{\win_i}$}
            \State \Return $(\specSEi)_{i\in\players }$
        \EndIf
        \State $\assumpi' \gets \assumpi \wedge\textcolor{blue}{\computeUCA}(\gamegraph,\assump_{-i}\wedge\speci,i)$ $\forall i\in\players$
        \If {$\assumpi' = \assumpi$ for all $i\in\players $}
        \State \Return $\false$
        \EndIf
        \State \Return $\recursiveGE(\game,(\assumpi')_{i\in\players })$
        \EndProcedure
	\end{algorithmic}

We have the following main result of this section.

\begin{theorem}
 Let $\game$ be a $k$-player game with game graph $\gamegraph = (V,E,v_0)$ and parity specifications $(\speci)_{i\in\players }$ such that $(\specSEi^*)_{i\in\players } = \ocomputeGE(\game)$, then $(\specSEi^*)_{i\in\players }$ is a GWSE for $\game$. Moreover, $\ocomputeGE$ terminates in time $\bigO(k^2\abs{E}\cdot(2\abs{V}+2\abs{E})^{d+2})$, where $d$ is the number of priorities used in the parity specifications.
\end{theorem}

\begin{proof}
Combining results from \cref{thm:computeGE} with \cref{thm:approxAPA,prop:computeUCA,prop:computeWin} gives us that $(\specSEi^*)_{i\in\players }$ is indeed a GWSE for $\game$.
Furthermore, as $\assumpi$ (for all $i\in\players $) in each iteration of the algorithm either remains the same or add more unsafe/colive edges, it can only change $2\abs{E}$ times.
Hence, as there are $k$ players, the algorithm $\ocomputeGE$ will terminate within $2k\abs{E}$ iterations.
Moreover, each iteration involves $k$ calls to both $\computeWin$ and $\computeUCA$.
Using Zielonka's algorithm\footnote{We note that the time complexity is exponential as we use Zielonka's algorithm~\cite{Zielonka98} to solve parity games. One can also use a quasi-polynomial algorithm~\cite{CaludeJKL017} for solving parity games to get a quasi-polynomial time complexity for $\ocomputeGE$.}~\cite{Zielonka98} for solving parity games, each iteration will take $\bigO((2\abs{V}+2\abs{E})^{d+2})$ time for $d$ priorities (by \cref{prop:computeWin,prop:computeUCA}).
In total, this gives us that $\ocomputeGE$ terminates in time $\bigO(k^2\abs{E}\cdot(2\abs{V}+2\abs{E})^{d+2})$.
\qed
\end{proof}

\begin{remark}
As Anand et al.\ show that APA's for games with \cobuchi specifications (i.e., $\spec = \LTLeventually\LTLalways T$ for some $T\subseteq V$) are always expressible by UCA's \cite[Thm. 3]{APA}, we note that $\computeAPA$ and $\approxAPA$ coincide for such games. This implies that no over approximation of assumptions is needed in this case an the optimizations discussed for $\computeUCA$ and $\computeWin$ can be directly applied for APA's. 

We further note that $\ocomputeGE$ also efficiently computes GWSE for games with more expressive specifications than \cobuchi. For instance, all games discussed in this paper as well as the mutual exclusion protocol discussed in \cite{ChatterjeeR14} can be solved by $\ocomputeGE$.
\end{remark}

\newpage
\bibliographystyle{splncs04}
\bibliography{main}

\begin{thebibliography}{10}
\providecommand{\url}[1]{\texttt{#1}}
\providecommand{\urlprefix}{URL }
\providecommand{\doi}[1]{https://doi.org/#1}

\bibitem{APA}
Anand, A., Mallik, K., Nayak, S.P., Schmuck, A.: Computing adequately
  permissive assumptions for synthesis. In: Sankaranarayanan, S., Sharygina, N.
  (eds.) Tools and Algorithms for the Construction and Analysis of Systems -
  29th International Conference, {TACAS} 2023, Held as Part of the European
  Joint Conferences on Theory and Practice of Software, {ETAPS} 2022, Paris,
  France, April 22-27, 2023, Proceedings, Part {II}. Lecture Notes in Computer
  Science, vol. 13994, pp. 211--228. Springer (2023).
  \doi{10.1007/978-3-031-30820-8\_15},
  \url{https://doi.org/10.1007/978-3-031-30820-8\_15}

\bibitem{AnandNS23}
Anand, A., Nayak, S.P., Schmuck, A.: Contract-based distributed synthesis in
  two-objective parity games. CoRR  \textbf{abs/2307.06212} (2023).
  \doi{10.48550/ARXIV.2307.06212},
  \url{https://doi.org/10.48550/arXiv.2307.06212}

\bibitem{baier2008principles}
Baier, C., Katoen, J.P.: Principles of model checking. MIT press (2008)

\bibitem{BloemCJK15}
Bloem, R., Chatterjee, K., Jacobs, S., K{\"{o}}nighofer, R.: Assume-guarantee
  synthesis for concurrent reactive programs with partial information. In:
  Baier, C., Tinelli, C. (eds.) Tools and Algorithms for the Construction and
  Analysis of Systems - 21st International Conference, {TACAS} 2015, Held as
  Part of the European Joint Conferences on Theory and Practice of Software,
  {ETAPS} 2015, London, UK, April 11-18, 2015. Proceedings. Lecture Notes in
  Computer Science, vol.~9035, pp. 517--532. Springer (2015).
  \doi{10.1007/978-3-662-46681-0\_50},
  \url{https://doi.org/10.1007/978-3-662-46681-0\_50}

\bibitem{admissible}
Brenguier, R., Raskin, J., Sankur, O.: Assume-admissible synthesis. Acta
  Informatica  \textbf{54}(1),  41--83 (2017). \doi{10.1007/s00236-016-0273-2},
  \url{https://doi.org/10.1007/s00236-016-0273-2}

\bibitem{BriceRB21}
Brice, L., Raskin, J., van~den Bogaard, M.: Subgame-perfect equilibria in
  mean-payoff games. In: Haddad, S., Varacca, D. (eds.) 32nd International
  Conference on Concurrency Theory, {CONCUR} 2021, August 24-27, 2021, Virtual
  Conference. LIPIcs, vol.~203, pp. 8:1--8:17. Schloss Dagstuhl -
  Leibniz-Zentrum f{\"{u}}r Informatik (2021).
  \doi{10.4230/LIPIcs.CONCUR.2021.8},
  \url{https://doi.org/10.4230/LIPIcs.CONCUR.2021.8}

\bibitem{BriceRB23}
Brice, L., Raskin, J., van~den Bogaard, M.: Rational verification for nash and
  subgame-perfect equilibria in graph games. In: Leroux, J., Lombardy, S.,
  Peleg, D. (eds.) 48th International Symposium on Mathematical Foundations of
  Computer Science, {MFCS} 2023, August 28 to September 1, 2023, Bordeaux,
  France. LIPIcs, vol.~272, pp. 26:1--26:15. Schloss Dagstuhl - Leibniz-Zentrum
  f{\"{u}}r Informatik (2023). \doi{10.4230/LIPIcs.MFCS.2023.26},
  \url{https://doi.org/10.4230/LIPIcs.MFCS.2023.26}

\bibitem{BrihayeBP14}
Brihaye, T., Bruy{\`{e}}re, V., {De Pril}, J.: On equilibria in quantitative
  games with reachability/safety objectives. Theory Comput. Syst.
  \textbf{54}(2),  150--189 (2014). \doi{10.1007/s00224-013-9495-7},
  \url{https://doi.org/10.1007/s00224-013-9495-7}

\bibitem{BruyereMR14}
Bruy{\`{e}}re, V., Meunier, N., Raskin, J.: Secure equilibria in weighted
  games. In: Henzinger, T.A., Miller, D. (eds.) Joint Meeting of the
  Twenty-Third {EACSL} Annual Conference on Computer Science Logic {(CSL)} and
  the Twenty-Ninth Annual {ACM/IEEE} Symposium on Logic in Computer Science
  (LICS), {CSL-LICS} '14, Vienna, Austria, July 14 - 18, 2014. pp. 26:1--26:26.
  {ACM} (2014). \doi{10.1145/2603088.2603109},
  \url{https://doi.org/10.1145/2603088.2603109}

\bibitem{BruyereRPR21}
Bruy{\`{e}}re, V., Roux, S.L., Pauly, A., Raskin, J.: On the existence of weak
  subgame perfect equilibria. Inf. Comput.  \textbf{276},  104553 (2021).
  \doi{10.1016/j.ic.2020.104553},
  \url{https://doi.org/10.1016/j.ic.2020.104553}

\bibitem{CaludeJKL017}
Calude, C.S., Jain, S., Khoussainov, B., Li, W., Stephan, F.: Deciding parity
  games in quasipolynomial time. In: Hatami, H., McKenzie, P., King, V. (eds.)
  Proceedings of the 49th Annual {ACM} {SIGACT} Symposium on Theory of
  Computing, {STOC} 2017, Montreal, QC, Canada, June 19-23, 2017. pp. 252--263.
  {ACM} (2017). \doi{10.1145/3055399.3055409},
  \url{https://doi.org/10.1145/3055399.3055409}

\bibitem{ChatterjeeDFR17}
Chatterjee, K., Doyen, L., Filiot, E., Raskin, J.: Doomsday equilibria for
  omega-regular games. Inf. Comput.  \textbf{254},  296--315 (2017).
  \doi{10.1016/j.ic.2016.10.012},
  \url{https://doi.org/10.1016/j.ic.2016.10.012}

\bibitem{chatterjee_SE}
Chatterjee, K., Henzinger, T.A., Jurdzinski, M.: Games with secure equilibria.
  Theor. Comput. Sci.  \textbf{365}(1-2),  67--82 (2006).
  \doi{10.1016/j.tcs.2006.07.032},
  \url{https://doi.org/10.1016/j.tcs.2006.07.032}

\bibitem{ChatterjeeHL10}
Chatterjee, K., Horn, F., L{\"{o}}ding, C.: Obliging games. In: Gastin, P.,
  Laroussinie, F. (eds.) {CONCUR} 2010 - Concurrency Theory, 21th International
  Conference, {CONCUR} 2010, Paris, France, August 31-September 3, 2010.
  Proceedings. Lecture Notes in Computer Science, vol.~6269, pp. 284--296.
  Springer (2010). \doi{10.1007/978-3-642-15375-4\_20},
  \url{https://doi.org/10.1007/978-3-642-15375-4\_20}

\bibitem{ChatterjeeR14}
Chatterjee, K., Raman, V.: Assume-guarantee synthesis for digital contract
  signing. Formal Aspects Comput.  \textbf{26}(4),  825--859 (2014).
  \doi{10.1007/s00165-013-0283-6},
  \url{https://doi.org/10.1007/s00165-013-0283-6}

\bibitem{DammF14}
Damm, W., Finkbeiner, B.: Automatic compositional synthesis of distributed
  systems. In: Jones, C.B., Pihlajasaari, P., Sun, J. (eds.) {FM} 2014: Formal
  Methods - 19th International Symposium, Singapore, May 12-16, 2014.
  Proceedings. Lecture Notes in Computer Science, vol.~8442, pp. 179--193.
  Springer (2014). \doi{10.1007/978-3-319-06410-9\_13},
  \url{https://doi.org/10.1007/978-3-319-06410-9\_13}

\bibitem{PrilFKSV14}
{De Pril}, J., Flesch, J., Kuipers, J., Schoenmakers, G., Vrieze, K.: Existence
  of secure equilibrium in multi-player games with perfect information. In:
  Csuhaj{-}Varj{\'{u}}, E., Dietzfelbinger, M., {\'{E}}sik, Z. (eds.)
  Mathematical Foundations of Computer Science 2014 - 39th International
  Symposium, {MFCS} 2014, Budapest, Hungary, August 25-29, 2014. Proceedings,
  Part {II}. Lecture Notes in Computer Science, vol.~8635, pp. 213--225.
  Springer (2014). \doi{10.1007/978-3-662-44465-8\_19},
  \url{https://doi.org/10.1007/978-3-662-44465-8\_19}

\bibitem{FiliotGR18}
Filiot, E., Gentilini, R., Raskin, J.: Rational synthesis under imperfect
  information. In: Dawar, A., Gr{\"{a}}del, E. (eds.) Proceedings of the 33rd
  Annual {ACM/IEEE} Symposium on Logic in Computer Science, {LICS} 2018,
  Oxford, UK, July 09-12, 2018. pp. 422--431. {ACM} (2018).
  \doi{10.1145/3209108.3209164}, \url{https://doi.org/10.1145/3209108.3209164}

\bibitem{FinkbeinerP22}
Finkbeiner, B., Passing, N.: Compositional synthesis of modular systems. Innov.
  Syst. Softw. Eng.  \textbf{18}(3),  455--469 (2022).
  \doi{10.1007/s11334-022-00450-w},
  \url{https://doi.org/10.1007/s11334-022-00450-w}

\bibitem{FismanKL10}
Fisman, D., Kupferman, O., Lustig, Y.: Rational synthesis. In: Esparza, J.,
  Majumdar, R. (eds.) Tools and Algorithms for the Construction and Analysis of
  Systems, 16th International Conference, {TACAS} 2010, Held as Part of the
  Joint European Conferences on Theory and Practice of Software, {ETAPS} 2010,
  Paphos, Cyprus, March 20-28, 2010. Proceedings. Lecture Notes in Computer
  Science, vol.~6015, pp. 190--204. Springer (2010).
  \doi{10.1007/978-3-642-12002-2\_16},
  \url{https://doi.org/10.1007/978-3-642-12002-2\_16}

\bibitem{KremerR03}
Kremer, S., Raskin, J.: A game-based verification of non-repudiation and fair
  exchange protocols. J. Comput. Secur.  \textbf{11}(3),  399--430 (2003).
  \doi{10.3233/jcs-2003-11307}, \url{https://doi.org/10.3233/jcs-2003-11307}

\bibitem{KupfermanS22}
Kupferman, O., Shenwald, N.: The complexity of {LTL} rational synthesis. In:
  Fisman, D., Rosu, G. (eds.) Tools and Algorithms for the Construction and
  Analysis of Systems - 28th International Conference, {TACAS} 2022, Held as
  Part of the European Joint Conferences on Theory and Practice of Software,
  {ETAPS} 2022, Munich, Germany, April 2-7, 2022, Proceedings, Part {I}.
  Lecture Notes in Computer Science, vol. 13243, pp. 25--45. Springer (2022).
  \doi{10.1007/978-3-030-99524-9\_2},
  \url{https://doi.org/10.1007/978-3-030-99524-9\_2}

\bibitem{LiLXHF14}
Li, X., Li, X., Xu, G., Hu, J., Feng, Z.: Formal analysis of fairness for
  optimistic multiparty contract signing protocol. J. Appl. Math.
  \textbf{2014},  983204:1--983204:10 (2014). \doi{10.1155/2014/983204},
  \url{https://doi.org/10.1155/2014/983204}

\bibitem{MajumdarMSZ20}
Majumdar, R., Mallik, K., Schmuck, A., Zufferey, D.: Assume-guarantee
  distributed synthesis. {IEEE} Trans. Comput. Aided Des. Integr. Circuits
  Syst.  \textbf{39}(11),  3215--3226 (2020). \doi{10.1109/TCAD.2020.3012641},
  \url{https://doi.org/10.1109/TCAD.2020.3012641}

\bibitem{MajumdarPS19}
Majumdar, R., Piterman, N., Schmuck, A.: Environmentally-friendly {GR(1)}
  synthesis. In: Vojnar, T., Zhang, L. (eds.) Tools and Algorithms for the
  Construction and Analysis of Systems - 25th International Conference, {TACAS}
  2019, Held as Part of the European Joint Conferences on Theory and Practice
  of Software, {ETAPS} 2019, Prague, Czech Republic, April 6-11, 2019,
  Proceedings, Part {II}. Lecture Notes in Computer Science, vol. 11428, pp.
  229--246. Springer (2019). \doi{10.1007/978-3-030-17465-1\_13},
  \url{https://doi.org/10.1007/978-3-030-17465-1\_13}

\bibitem{vmcai}
Schmuck, A.K., Thejaswini, K.S., Sa{\u{g}}lam, I., Nayak, S.P.: Solving
  two-player games under progress assumptions. In: Dimitrova, R., Lahav, O.,
  Wolff, S. (eds.) Verification, Model Checking, and Abstract Interpretation.
  pp. 208--231. Springer Nature Switzerland, Cham (2024)

\bibitem{Steg22}
Steg, J.: On identifying subgame-perfect equilibrium outcomes for timing games.
  Games Econ. Behav.  \textbf{135},  74--78 (2022).
  \doi{10.1016/j.geb.2022.05.012},
  \url{https://doi.org/10.1016/j.geb.2022.05.012}

\bibitem{Ummels06}
Ummels, M.: Rational behaviour and strategy construction in infinite
  multiplayer games. In: Arun{-}Kumar, S., Garg, N. (eds.) {FSTTCS} 2006:
  Foundations of Software Technology and Theoretical Computer Science, 26th
  International Conference, Kolkata, India, December 13-15, 2006, Proceedings.
  Lecture Notes in Computer Science, vol.~4337, pp. 212--223. Springer (2006).
  \doi{10.1007/11944836\_21}, \url{https://doi.org/10.1007/11944836\_21}

\bibitem{Zielonka98}
Zielonka, W.: Infinite games on finitely coloured graphs with applications to
  automata on infinite trees. Theor. Comput. Sci.  \textbf{200}(1-2),  135--183
  (1998). \doi{10.1016/S0304-3975(98)00009-7},
  \url{https://doi.org/10.1016/S0304-3975(98)00009-7}

\end{thebibliography}

\newpage
\appendix
\section{Correctness of \computeUCA}\label{appendix:UCA}
Let us first prove some results in the following lemma that is needed to prove \cref{prop:computeUCA}.
\begin{lemma}\label{lemma:graphUnsafeColive}
    Given a game graph $\gamegraph$ with parity specification $\spec$ and an UCA $\assump = \assumptemp{\safegroup}{\colivegroup}$ for $\playersnoti$, let $\hat{\gamegraph},\hat{\spec}$ be the game graph and parity specification, respectively computed in \cref{algo:computeUCA}. Then, the following holds: 
    \begin{enumerate}
        \item $\lang(\hat{\gamegraph},\hat{\spec})|_V = \lang(\gamegraph,\spec\wedge\assump)$;\label{item:lemma:graphUnsafeColive1}
        \item Every UCA $\hat{\assump} = \assumptemp{\hat{S}}{\hat{C}}$ for $\p{1}$ in $\hat{\gamegraph}$ is an UCA for $\p{i}$ in $\gamegraph$.\label{item:lemma:graphUnsafeColive2} 
    \end{enumerate}
\end{lemma}
\begin{proof}
        It is easy to see that, for any play $\play$ in $\hat{\gamegraph}$, the play $\play|_V$, obtained by restricting the vertices of $\play$ only to $V$, is the corresponding play in $\gamegraph$. 
        Now, let us prove both items of the lemma.
        \paragraph*{(1)} Let $\play$ be play in $\lang(\gamegraph,\spec\wedge\assump)$.
        Then, $\play\vDash\assumpsafe(\safegroup)$, and hence, no unsafe edge of $\safegroup$ appear in $\play$. 
        So, by construction as all other edges are retained from $\gamegraph$ to $\hat{\gamegraph}$ (with additional middle vertex for colive edges), there exists a corresponding play $\hat{\play}$ in $\hat{\gamegraph}$ such that $\hat{\play}|_V = \play$.
        Furthermore, as colive edges appear only finitely often in $\play$ and $\play$ satisfies parity specification $\spec$, priority $2d+1$ only appear finitely often in $\hat{\play}$.
        Moreover, by the parity condition, the maximal priority in $[0;2d]$ appearing infinitely in $\hat{\play}$ is even.
        So, $\hat{\play}$ also satisfies parity specification $\hat{\spec}$, and hence, $\hat{\play}\in\lang(\hat{\gamegraph},\hat{\spec})$.
        Analogously, one can show the other direction by proving for every play $\hat{\play}\in\lang(\hat{\gamegraph},\hat{\spec})$, the corresponding play $\hat{\play}|_V$ belongs to $\lang(\gamegraph,\spec\wedge\assump)$.
        \paragraph*{(2)} As $\hat{\assump}$ is an UCA for $\p{1}$, it holds that $\hat{S},\hat{C}\subseteq \hat{E_1} = E_i$. Hence, $\hat{\assump}$ is also an (well-defined) UCA for $\p{i}$ in $\gamegraph$.
    \qed
    \end{proof}

Now, let us prove \cref{prop:computeUCA}. 
\restateUCA*
\begin{proof}
    By \cref{lemma:graphUnsafeColive}, $\hat{\spec}$ expresses the same language as $\assump\wedge\spec$. Moreover, as $E_i = \hat{E_1}$, both APA's computed by $\computeAPA(\gamegraph,\assump\wedge\spec,i)$ and $\computeAPA(\hat{\gamegraph},\hat{\spec},1)$ uses the same edges of $\p{i}$ in each template. Hence, approximating them with $\approxAPA$ results in same unsafe edge set and colive edge set. Therefore, the UCA's $\approxAPA(\hat{\gamegraph},\hat{\spec},1) = \approxAPA(\gamegraph,\assump\wedge\spec,i)$.
    As $\computeUCA$ returns the UCA $\approxAPA(\hat{\gamegraph},\hat{\spec},1)$, we have $\lang(\assump') = \lang(\assump'')$.
    
    Furthermore, by construction, $\abs{\hat{V}} \leq \abs{V}+\abs{E}$. Hence, by \cref{lemma:computeAPA}, $\computeUCA$ terminates in time $\bigO((\abs{V}+\abs{E})^4)$.
\qed
\end{proof}

\section{Procedure $\computeWin$ and its Correctness}\label{sec:app:win}

The procedure $\computeWin$ is given in \cref{algo:computeWin}, that solves a game $(\gamegraph,\specSE)$ for $\p{i}$ where $\specSE = \assump_1\wedge (\assump_2\Rightarrow\spec)$ with UCA $\assump_1$ for $\p{i}$ and UCA $\assump_2$ for $\playersnoti$. 
Intuitively, as we only need to satisfy $\assump_2$ or $\spec$, once an unsafe-edge of $\assump_2$ is used, we only need to satisfy $\assump_1$.
Hence, we have two copies of game graph: one copy to encode both assumptions and one identical to the game graph of \cref{algo:computeUCA} only to ensure $\assump_1$, both connected by vertices representing unsafe-edges of $\assump_2$.

\begin{algorithm}
	\caption{$\computeWin(\gamegraph,\assumptemp{\safegroup_1}{\colivegroup_1}\wedge(\assumptemp{\safegroup_2}{\colivegroup_2}\Rightarrow\spec),i)$}\label{algo:computeWin}
	\begin{algorithmic}[1]
		\Require A $k$-player game graph $\gamegraph = (V,E,v_0)$, a specification $\specSE = \assump_1\wedge(\assump_2\Rightarrow\spec)$ with UCA $\assump_1 = \assumptemp{\safegroup_1}{\colivegroup_1}$ for $\p{i}$ and UCA $\assump_2 = \assumptemp{\safegroup_2}{\colivegroup_2}$ for $\playersnoti$, and parity specification $\spec=\parity(\priority)$ s.t.\ $\priority:V\rightarrow [0;2d]$.
		\Ensure The set $\team{i}(\gamegraph,\specSE)$.
        \State $\hat{V}_1 \gets V_i\cup \{\hat{v}\mid v\in V_i\}$ and $\hat{V}_2 \gets V_{-i}\uplus (\colivegroup_1\cup\colivegroup_2\cup\safegroup_2) \cup \{\hat{v}\mid v\in V_{-i}\uplus \colivegroup_1\}$\label{algo:computeWin:graph1}
        \State $\begin{aligned}[t] \hat{E}
            &\gets E\setminus(\safegroup_1\cup\safegroup_2\cup\colivegroup_1\cup\colivegroup_2)\cup \{(\hat{u},\hat{v}) \mid (u,v)\in E\setminus (S_1\cup C_1)\}\\
            &\cup \{(u,c),(c,v) \mid c = (u,v)\in\colivegroup_1\cup\colivegroup_2\}
            \cup \{(\hat{u},\hat{c}),(\hat{c},\hat{v}) \mid c = (u,v)\in\colivegroup_1\}\\
            &\cup \{(u,s), (s,\hat{v})\mid s = (u,v) \in\safegroup_2\}
        \end{aligned}$\label{algo:computeWin:graph2}
        \State $\hat{\priority} = \begin{cases}
            \priority(v) &\text{if $v\in V$}\\
            2d &\text{if $v\in \colivegroup_2$}\\
            2d+1 &\text{if $v = u$ or $\hat{u}$ for some $u\in \colivegroup_1$}\\
            0 &\text{otherwise}.
        \end{cases}$\label{algo:computeWin:parity}
        \State $\hat{\gamegraph} = (\hat{V}_1\uplus \hat{V}_2,\hat{E},v_0)$; $\hat{\spec}\gets\parity(\hat{\priority})$\label{algo:computeWin:graph3}
        \State \Return $V\cap \team{1}(\hat{\gamegraph},\hat{\spec})$\label{algo:computeWin:compute}
	\end{algorithmic}
\end{algorithm}

Similar to the construction presented in the previous section, due to the simplicity and locality of UCA's, the modifications required to go from $(\gamegraph,\specSE)$ to $(\hat{\gamegraph},\hat{\spec})$ (\cref{algo:computeWin:graph1,algo:computeWin:graph2} in \cref{algo:computeWin}) are mild enough that the check required in \cref{algo:computeGE:if1} of \cref{algo:computeGE} can be directly performed over $\hat{G}$ without loosing correctness, as shown in \cref{sec:ocomputeGE}.

\begin{restatable}{proposition}{restatecomputeWin}\label{prop:computeWin}
    Given a game graph $\gamegraph = (V,E,v_0)$ with parity specification $\spec$, UCA $\assump_1 = \assumptemp{\safegroup_1}{\colivegroup_1}$ for $\p{i}$ and UCA $\assump_2 = \assumptemp{\safegroup_2}{\colivegroup_2}$ for $\playersnoti$, let $\win = \computeWin(\gamegraph,\specSE,i)$ with $\specSE = \assump_1\wedge(\assump_2\Rightarrow\spec)$. Then, it holds that $v_0\in\team{i}(\gamegraph,\specSE)$ if and only if $v_0\in\win$.
    Furthermore, $\computeWin$ for $\spec$ with $d$ priorities terminates in time $\bigO((2\abs{V}+2\abs{E})^{d+2})$.
\end{restatable}

\begin{proof}
Let $v_0\in\team{i}(\gamegraph,\specSE)$. Then, $\p{i}$ has a strategy $\strat$ such that $\strat\vDash\specSE$.
So, $\strat\vDash\assump_1$, and hence, $\strat$ never uses any unsafe edge of $\safegroup_1$.
By construction, we have $E_i\setminus \safegroup_1 = \hat{E_1}$. Hence, the following strategy $\hat{\strat}$ of $\p{1}$ in $\hat{\gamegraph}$ is well-defined: $\hat{\strat}(\hat{\play}) = \strat(\hat{\play}|_V)$ for each $\play\in \hat{V}^*\hat{V_1}$.
It is enough to show that $\hat{\strat}\vDash\hat{\spec}$.
Let $\hat{\play}$ be a $\hat{\strat}$-play.
Let $\play = \hat{\play}|_V$ be the corresponding $\strat$-play.
As $\strat\vDash\assump_1$, $\play$ uses edges of $C_1$ finitely many time, hence, $\hat{\play}$ visits vertices in $C_1$, i.e., ones with priority $2d+1$, finitely often.
Moreover, as $\strat\vDash \neg\assump_2\vee\spec$, it holds that $\play\vDash\neg\assump_2$ or $\play\vDash\spec$.
Let $\play\vDash\neg\assump_2$. If $\play$ uses an edge $s\in S_2$, then $\hat{\play}$ move to the second copy, i.e., $\{\hat{v}\mid v\in V\uplus C_1\}$, via $s$, and hence, only visits priority $0$ infinitely often (as it does not visit priority $2d+1$ infinitely often). So, $\hat{\play}$ satisfies $\hat{\spec}$.
If $\play$ uses an edge $c\in C_2$ infinitely often, then $\hat{\play}$ visits $c$ with priority $2d$ infinitely often, and hence, satisfies $\hat{\spec}$.
Finally, now, suppose $\play\vDash\spec$, then, the maximum priority $\play$ visits infinitely often in $[0;2d]$ is even. As $\hat{\play}$ does not visits priority $2d+1$ infinitely often, the maximum priority it visits is also even.
Therefore, in any case, $\hat{\play}$ satisfies $\hat{\spec}$, and hence, $\hat{\strat}\vDash\hat{\spec}$.
Analogously, using similar arguments, one can show that if $v_0\in\team{1}(\hat{\gamegraph},\hat{\spec})$, then $v_0\in\team{i}(\gamegraph,\specSE)$.

Furthermore, by construction, $\abs{\hat{V}}\leq 2(\abs{V}+\abs{E})$ as there are two copies each with at most $(\abs{V}+\abs{E})$ vertices. Then, the time complexity follows from Zielonka's algorithm's time complexity.
\qed
\end{proof}

\end{document}